\documentclass[lettersize,journal]{IEEEtran}
\IEEEoverridecommandlockouts
\usepackage{mathrsfs}
\usepackage{amsmath}
\usepackage{amssymb}
\usepackage{amsthm}
\usepackage{setspace}
\usepackage{tabularx}
\ifCLASSINFOpdf \else \fi

\usepackage{graphicx}
\usepackage{algorithm}
\usepackage{algorithmic}
\usepackage[numbers,sort&compress]{natbib}
\usepackage{color}
\newtheorem{Theo}{Theorem}
\newtheorem{lemma}{Lemma}

\usepackage{amsmath,amsfonts}
\usepackage{array}
\usepackage[caption=false,font=normalsize,labelfont=sf,textfont=sf]{subfig}
\usepackage{textcomp}
\usepackage{stfloats}
\usepackage{url}
\usepackage{verbatim}
\begin{document}

\title{Spectrum Sharing Towards Delay Deterministic Wireless Network: Delay Performance Analysis}

\author{Zhiqing~Wei,~\IEEEmembership{Member,~IEEE,}
        Ling~Zhang,~\IEEEmembership{Student Member,~IEEE,}\\
        Gaofeng~Nie,~\IEEEmembership{Member,~IEEE,}
        Huici~Wu,~\IEEEmembership{Member,~IEEE,}
        Ning~Zhang,~\IEEEmembership{Senior Member,~IEEE,}\\
        Zeyang~Meng,~\IEEEmembership{Student Member,~IEEE,}
        and~Zhiyong Feng,~\IEEEmembership{Senior Member,~IEEE}
\thanks{
This work was supported in part by the National Natural Science Foundation of China (NSFC) under Grant 62271081, 92267202 and U21B2014, in part by the National Key Research and Development Program of China under Grant 2020YFA0711303, and in part by the Young Elite Scientists Sponsorship Program by CAST under Grant 2020QNRC001.

Zhiqing Wei, Ling Zhang, Gaofeng Nie, Zeyang Meng, and Zhiyong Feng
are with the School of Information and Communication Engineering, Beijing University of Posts and Telecommunications, Beijing 100876, China (email: \{weizhiqing; zhangling\_zl; niegaofeng; mengzeyang; fengzy\}@bupt.edu.cn).

Huici Wu is with the National Engineering Research Center of Mobile Network Technologies, Beijing University of Posts and Telecommunications, Beijing 100876, China (e-mail: dailywu@bupt.edu.cn).

Ning Zhang is with the Department of Electrical and Computer Engineering, University of Windsor, Windsor, ON, N9B 3P4, Canada. (e-mail: ning.zhang@uwindsor.ca).

Correspondence authors: Zhiyong Feng, Huici Wu.
}}

\maketitle

\begin{abstract}
To accommodate Machine-type Communication (MTC) service, the wireless network needs to support low-delay and low-jitter data transmission, realizing delay deterministic wireless network. This paper analyzes the delay and jitter of the wireless network with and without spectrum sharing. When sharing the spectrum of the licensed network, the spectrum band of wireless network can be expanded, such that the delay and jitter of data transmission are reduced. The challenge of this research is to model the relation between the delay/jitter and the parameters such as node distribution, transmit power, and bandwidth, etc. To this end, this paper applies stochastic geometry and queueing theory to analyze the outage probability of the licensed network and the delay performance of the wireless network with and without spectrum sharing. By establishing the M/G/1 queueing model for the queueing of the Base Station (BS) in the wireless network, the downlink delay and jitter are derived. Monte Carlo simulation results show that the spectrum sharing reduces the delay and jitter without causing serious interference to the licensed network, which can lay a foundation for the application of spectrum sharing in delay deterministic wireless network supporting MTC service.
\end{abstract}


\section{Introduction}
\IEEEPARstart{W}{ith} the accelerated penetration of wireless technology into the industrial field, the performance of industrial wireless network has been continuously improved, and the intelligent manufacturing industry represented by process industry and automatic control has developed rapidly. For example, in a smart factory, the computing system collects data from sensors and machines to make decisions, and the machines perform corresponding operations according to the received control information \cite{7123559}. In the above scenario, in order to ensure real-time data transmission and control of machines, industrial wireless network needs to support low-delay and low-jitter communication services.

\indent In the recent literature, resource allocation \cite{5764568,9606568,9560065,9725256,10018944} and scheduling methods \cite{8489935,8964351,8882510} are studied to guarantee the low-delay and low-jitter communication performance. \textit{On resource allocation methods}, Maadani \textit{et al.} \cite{5764568} analyzed the priority of IEEE 802.11e from the perspective of packet delay and reliability. Then, they studied the superiority of using spatial multiplexing technology to reduce delay in Multiple-input Multiple-output (MIMO) based industrial wireless network. Jin \textit{et al.} introduced Mobile Edge Computing (MEC) into Industrial Internet of Things (IIoT) applications in \cite{9606568} and proposed an Improved Differential Evolution (IDE) algorithm to optimize the traffic offloading and resource allocation to minimize the weighted sum of delay and energy in the multi-hop MEC-based IIoT nerwork. In \cite{9560065}, combining Non-orthogonal Multiple Access (NOMA) and MEC techniques, a joint optimization problem of sub-channel allocation, offloading decision and computational resource allocation is studied to minimize the average task delay of all users. In \cite{10018944}, Zhang \textit{et al.} proposed a spectrum allocation optimization mechanism based on federated learning, which minimizes the delay of federated learning to improve the transmission efficiency of mobile networks, taking into account the energy consumption of a single participating device. \textit{On scheduling methods}, Chen \textit{et al.} \cite{8489935} represented the scheduling and channel allocation solutions as flow-link-channel-slot tuples and minimized the end-to-end delay of the network through joint transmission scheduling and channel allocation. In order to ensure that the real-time data flows arrive at the destination devices under the deadline constraint and reduce the average transmission delay, Wang \textit{et al.} \cite{8964351} proposed the link conflict classification based branch and bound algorithm that obtains the optimal scheduling solution by constructing a search tree. Besides, the least conflict degree algorithm is proposed that allocates channels by dynamically adjusting the order of data streams in a heuristic manner. Simulation results verify that both algorithms reduce the average transmission delay in Industrial Wireless Sensor Networks (IWSNs). In \cite{8882510}, Li \textit{et al.} proposed a channel-based Optimal Back-off Delay Control (OBDC) scheme, which minimizes the total time of packet transmission in IWSNs. By combining Level Crossing Rate (LCR) of the received signal in the industrial wireless network and the proposed phase-type semi-Markov model, the OBDC scheme can react quickly to time-varying wireless channels and find the optimal back-off delay for each packet to control transmission delay when the number of nodes and the LCR are changing.

\indent Overall, the delay and jitter in communication can be reduced by allocating or scheduling resources.
However, with the expansion of industrial equipment,
only relying on resource allocation and scheduling methods can not meet high communication requirements.
The bandwidth requirements for ensuring real-time transmission of large amount of data are increasing \cite{8825819}.
And the limited available spectrum resources have become the bottleneck limiting the delay performance of the wireless network \cite{6883319}.
In order to solve the contradiction between the promotion of industrial equipment and limited spectrum resources,
communication in industrial wireless network need to expand the spectrum to enhance system capacity \cite{9063521}.
However, it is unlikely to obtain enough licensed frequency bands for industrial wireless network \cite{9063521}.
Due to the scarcity of spectrum resources and the high cost of exploiting new frequency bands, spectrum sharing has become promising to improve spectrum utilization.
As a secondary network,
the industrial wireless network serving machine-type communications accesses the spectrum used by the primary network, so as to expand its spectrum,
improving network performance.

\indent There are a large number of studies on spectrum sharing in industrial wireless network \cite{9520317}.
Rodriguez \textit{et al.} evaluated the network performance based on cognitive radio in different industrial scenarios in \cite{7160576}.
Research shows that cognitive radio has the potential to maintain the performance of network in harsh channels and under interference.
Si \textit{et al.} applied cognitive radio in \cite{8254745} to obtain enough frequency bands for delay tolerant industrial data transmission and applied edge computing to generate real-time responses to delay sensitive industrial data transmission. When sharing spectrum, users will receive multiple different signals. it is necessary to identify the signals to obtain useful signals.
Liu \textit{et al.} proposed a signal classification framework based on deep learning networks in \cite{9063521}. And various signal classification methods have been proposed to adapt to general situations.
Considering the different energy consumption when cognitive radio equipment switches to different frequency bands, Demirci \textit{et al.}
in \cite{8887230} proposed a polynomial time heuristic algorithm to solve the energy consumption problem caused by channel switching.
However, due to spectrum switching in this scenario,
communication may even be interrupted.

\indent Some researchers have analyzed the delay performance of spectrum sharing. In \cite{7707414}, Sibomana \textit{et al.} put forward the transmit power strategy of 
secondary users. And the expressions of the average packet waiting delay in the queue and the total delay in the system for each class of traffic are obtained. In \cite{7589870}, Bayrakdar \textit{et al.} studied the end-to-end delay of cognitive radio networks based on Time Division Multiple Access (TDMA). In \cite{9306105}, Amini \textit{et al.} analyzed the average packet delay of the cognitive IoT network and studied the impact of primary network traffic behaviors and loT network parameters on delay of the IoT network. However, in \cite{7589870} and \cite{9306105}, the secondary user cannot always use the spectrum to ensure the real-time data transmission in industrial wireless network. \cite{7746101} and \cite{8785815} both analyze the delay performance when using competitive access to shared spectrum, but spectrum sharing based on competition will lead to communication interruption, which is not allowed in industrial wireless network. Lin \textit{et al.} \cite{6883319,7297795} designed autonomous switching rules to improve the spectral efficiency of IWSN on the Industrial, Scientific and Medical (ISM) band. However, spectrum switching will cause service interruption, which can affect the real-time data transmission.

\indent In order to solve the problems that spectrum switching and low priority of users in previous studies cannot guarantee real-time communication in industrial wireless network, 
this paper proposes a spectrum sharing model. The licensed network is a Human-type Communication (HTC) network. The wireless network supporting MTC service is a MTC network, which shares the spectrum of HTC network. Through the analysis of delay and jitter of MTC network, the impact of the parameters such as transmit power, density of users, etc., on the delay and jitter is revealed, which provides a guideline for the delay and jitter optimization of MTC network with spectrum sharing.
The contributions of this paper are as follows.

\begin{enumerate}
\item{The spectrum sharing model is established by
using stochastic geometry. The outage probability of HTC network
    with and without spectrum sharing is analyzed,
    providing a theoretical basis for reducing
    the interference of MTC network to HTC network.}
\item{The M/G/1 queueing model is established for the BS in MTC network.
The expressions of delay and delay jitter with and without spectrum sharing are obtained when MTC network has  proprietary spectrum or not.
Simulation analysis shows that spectrum sharing can reduce the delay and delay jitter of MTC network when MTC network has proprietary spectrum, 
which provides a solution for
improving the delay determinism of MTC networks.}
\end{enumerate}

\indent The structure of this paper is as follows.
Section II introduces the system model.
Section III analyzes the performance of HTC networks with and without
spectrum sharing, and expressions for the outage probability are obtained.
Section IV analyzes the performance metrics of MTC network
such as the delay and jitter with or without proprietary spectrum and obtains the delay performance.
The simulation results and analysis are given in Section V,
which verify the accuracy of the theoretical derivation by comparing Monte Carlo simulation results and theoretical results.
Section VI provides the conclusion of this research.
The variables and acronyms in this paper are provided in Table \ref{Table_1}.
\begin{table}[!t]\footnotesize
\label{Table_1}
 \caption{The Key Notations}
 \renewcommand{\arraystretch}{1} 
 \begin{center}
  \begin{tabular}{m{0.07\textwidth}m{0.4\textwidth}}
   \hline
   \hline
   Notation & Description \\
   \hline
   $ P_h $& Transmit power of HBSs \\
   $ P_m $ & Transmit power of MBS when using proprietary spectrum \\
   ${P_{{\rm{m}}}^{{\rm{'}}}}$ & Transmit power of MBS when using shared spectrum\\
   $ x_0 $ & Distance from the typical UE to the closest HBS \\
   $ y_0 $ & Distance from the typical MTC device to MBS \\
   $ B_h $ & Bandwidth of shared spectrum \\
   $ B_m $ & Bandwidth of the proprietary spectrum of MTC network \\
   $ N $ & Power spectral density of noise \\
   $ \alpha $ & Path loss factor \\
   $ U_m $ & Size of a data packet \\
   $ t_{out} $ & Longest service delay \\
   $\lambda_h $ & Distribution density of HBSs \\
   $\lambda_{mu} $ & Distribution density of MTC devices \\
   $ N_h $ & Number of UEs served by each HBS \\
   $ N_m $ & Number of MTC devices served by each MBS \\
   MTC  & Machine-type Communication  \\
   BS  & Base Station  \\
   MIMO  & Multiple-input Multiple-output  \\
   MEC  & Mobile Edge Computing  \\
   IIoT  & Industrial Internet of Things  \\
   IDE  & Improved Differential Evolution  \\
   NOMA  & Non-orthogonal Multiple Access  \\
   IWSN  & Industrial Wireless Sensor Network \\
   OBDC  & Optimal Back-off Delay Control  \\
   LCR  &  Level Crossing Rate \\
   TDMA &  Time Division Multiple Access \\
   ISM  & Industrial, Scientific and Medical  \\
   HTC  & Human-type Communication  \\
   MBS &  BS in MTC network  \\
   HBS &  BS in HTC network  \\
   PPP  & Poisson Point Process  \\
   UE   & User Equipment \\
   SINR & Signal to Interference plus Noise Ratio\\
   SNR & Signal to Noise Ratio\\
   CDF  & Cumulative Distribution Function  \\
   PDF  & Probability Density Function  \\
   FCFS  & First Come First Served  \\
   LTE-U & LTE-Unlicensed\\
   5G NR-U  & 5G New Radio in Unlicensed Spectrum \\
   \hline
   \hline
  \end{tabular}
 \end{center}
\end{table}

\section{System Model}

\begin{figure}
\centering
\includegraphics[width=0.49\textwidth]{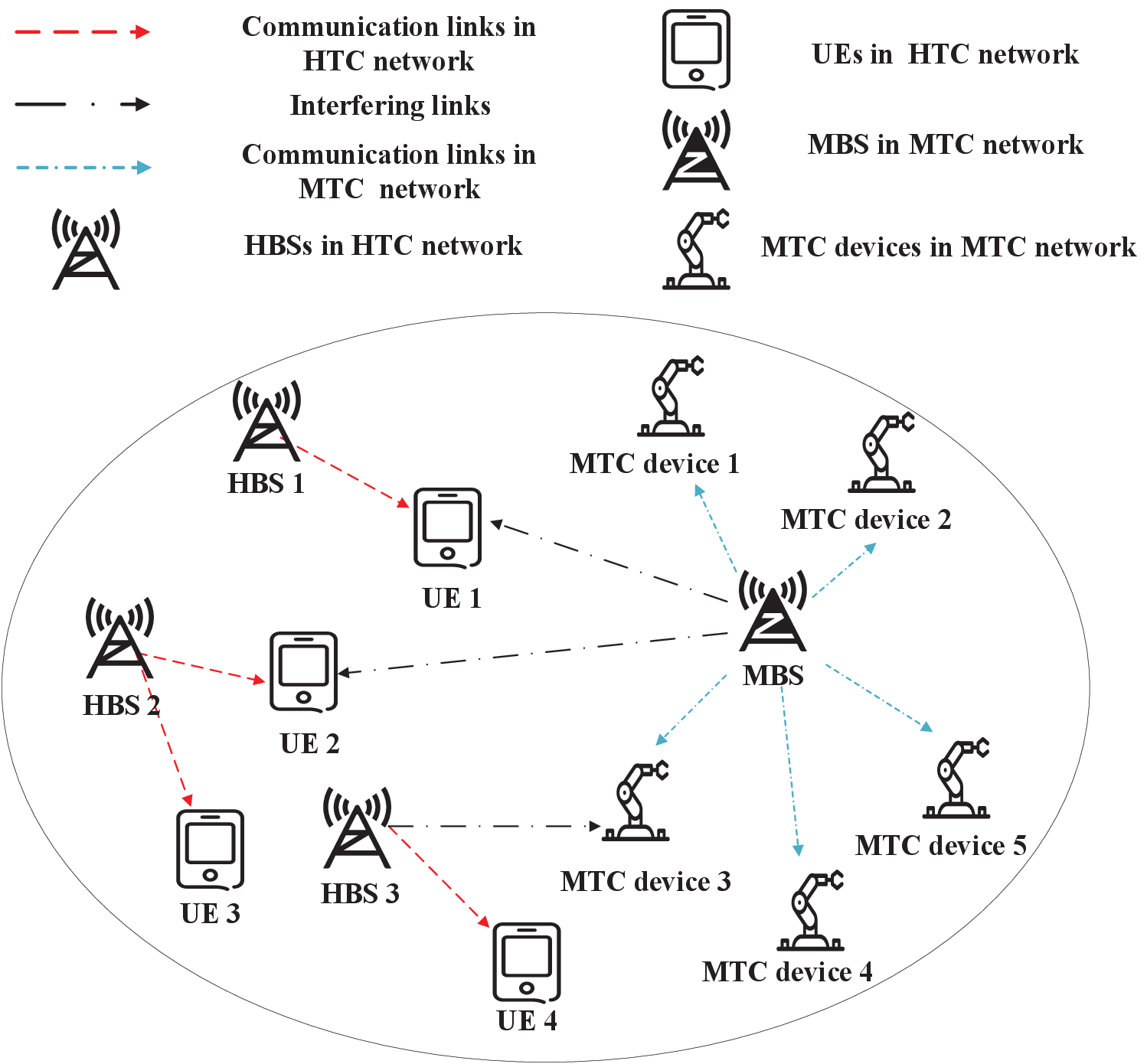}
\caption{System model of spectrum sharing between MTC network and HTC network.}
\label{fig_system_model}
\end{figure}

\begin{figure}
\centering
\includegraphics[width=0.35\textwidth]{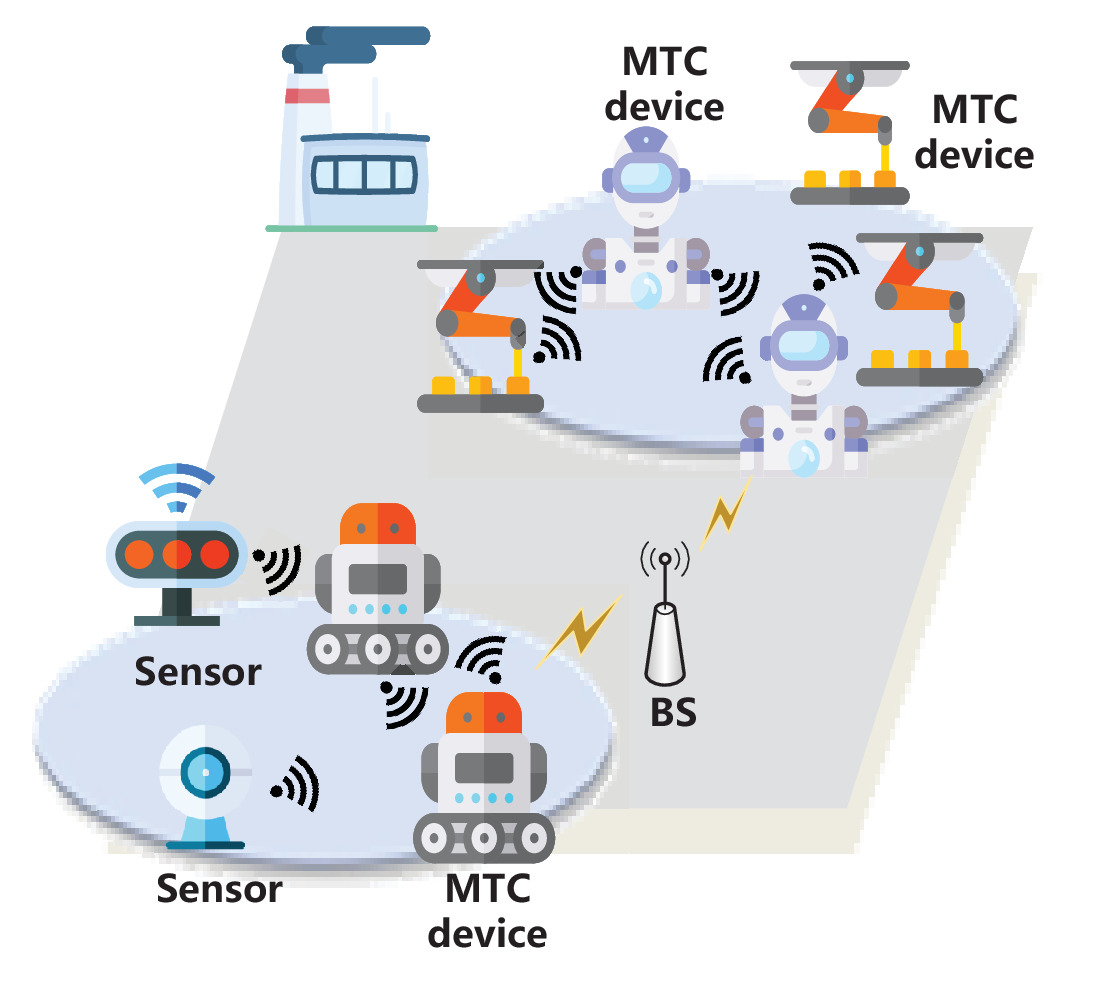}
\caption{Machine-type communication in a closed workshop.}
\label{fig_industrial}
\end{figure}

As shown in Fig. \ref{fig_system_model}, the system model considers two types of networks, namely MTC network and HTC network.
MTC network shares the spectrum of HTC network as long as it causes a small amount of interference to HTC network.
The spectrum bandwidth of HTC network is expressed as ${B_h}$.
The proprietary spectrum bandwidth of MTC network is denoted as ${B_m}$.
There are multiple BSs in HTC network, denoted by HTC BSs (HBSs).
The set of HBSs can be expressed as $\phi _h$ and the location distribution of HBSs follows a homogeneous Poisson Point Process
(PPP) with density $\lambda_h$.
Each HBS serves $N_h$ UEs, and UEs are uniformly distributed within the service area of HBS. As shown in  Fig. \ref{fig_industrial},
only one BS is deployed to provide communication services in a closed workshop. All MTC devices communicate through this BS.
Therefore, the MTC network in Fig. \ref{fig_system_model} consists of one MTC BS (MBS) and $N_m$ MTC devices,
where the MBS adopts the shared spectrum to serve the MTC devices.
The set of MTC devices can be represented as $\phi _m $.
The location distribution of MTC devices follows a homogeneous PPP with density $\lambda_{mu}$.

\indent Since UEs and MTC devices follow homogeneous PPP, according to the Slyvniak’s theorem \cite{stoyan2013stochastic}, the outage probability of HTC network and the delay deterministic feature of MTC network can be evaluated by analyzing the performance of the typical UE and the typical MTC device. We assume that the typical UE and the typical MTC device are distributed at the origin of the two-dimensional plane where the HTC and MTC networks are deployed  \cite{9115898,7756327}. It is assumed that the small-scale fading of the communication channel is Rayleigh fading, namely the channel gain obeys a negative exponential distribution with unit mean.

\section{Performance of HTC Network}
For HTC network, the success probability of communication is a crucial performance metric. It is defined as the probability that the Signal to Interference plus Noise Ratio (SINR) at the receiver is greater than a threshold \cite{7842290,9261468}. On the contrary, the outage probability is the probability that the SINR of the receiver is smaller than the threshold.

\subsection{Outage probability without spectrum sharing}
Firstly, the performance of HTC network without spectrum sharing is analyzed. Spectrum is evenly allocated to guarantee simultaneous communication between HBS and ${N_h}$ UEs. The SINR of the typical UE is expressed as follows.
\begin{equation}\label{exp1}
{\gamma _h}= \frac{{{P_h}x_0^{ - \alpha }{h_0}}}{{\sum\limits_{i \in {\phi _h}\backslash \{ 0\} } {{P_h}x_i^{ - \alpha }{h_i} + \frac{{N{B_h}}}{{{N_h}}}} }},\\
\end{equation}
where $P_h$ is the transmit power of HBSs. $x_0$ is the distance between the typical UE and C.
${x_i}$ is the distance between the typical UE and the $i$th HBS. ${h_i}$ is the channel gain of the typical UE communicating with the $i$th HBS. $N$ is the power spectral density of noise. $\alpha$ is the path loss factor. \\ \indent According to the definition, the success probability of communication $P_{suc}$ for the typical UE is
\begin{equation}\label{exp2}
{P_{suc}} = P( {{\gamma _h} > {\theta _h}}),
\end{equation}
where $\theta_h$ is the threshold of SINR. Therefore, the outage probability of communication $P_{out}$ is
\begin{equation}\label{exp3}
{P_{out}} = 1 - P( {{\gamma _h} > {\theta _h}}).
\end{equation}
\indent Then, we have Lemma 1.

\begin{lemma}\label{lemma1}
The outage probability of communication for the typical UE in HTC network without spectrum sharing is \cite{7842290}
\begin{equation}\label{exp4}
{P_{out}} = 1 - \exp ({ - x_0^\alpha {\theta _h}\frac{{N{B_h}}}{{{N_h}{P_h}}} - {\lambda _h}\frac{{2{\pi ^2}x_0^2{\theta _h}^{\frac{2}{\alpha }}}}{{\alpha \sin ( {\frac{{2\pi }}{\alpha }} )}}} ).
\end{equation}
\end{lemma}
\begin{proof}
Substituting (\ref{exp1}) into (9) in \cite{7842290}, the success probability of communication is
\begin{equation}\label{exp5}
\begin{aligned}
{P_{suc}} = P(  \frac{{{P_h}x_0^{ - \alpha }{h_0}}}{{\sum\limits_{i \in {\phi _h}\backslash \{ 0\} } {{P_h}x_i^{ - \alpha }{h_i} + \frac{{N{B_h}}}{{{N_h}}}} }}> {\theta _h} )\\
 = \exp ( { - x_0^\alpha {\theta _h}\frac{{N{B_h}}}{{{N_h}{P_h}}}}) \cdot {L_{{I_h}}}( {x_0^\alpha {\theta _h}}),
\end{aligned}
\end{equation}
where ${L_{{I_h}}}( {x_0^\alpha {\theta _h}} )$ is the Laplace transform of ${I_h}$. ${I_h}$ is expressed as
\begin {equation}\label{exp6}
 {I_h} = \sum\limits_{i \in {\phi _h}\backslash \{ 0\} } {x_i^{ - \alpha }{h_i}} .
\end {equation}
\indent According to \cite{7842290}, ${L_{{I_h}}}( {x_0^\alpha {\theta _h}} )$ is expressed as
\begin{equation}\label{exp7}
\begin{aligned}
&{L_{{I_h}}}( {x_0^\alpha {\theta _h}} )\\
&= \exp ( { - \int_0^{2\pi } {\int_0^\infty  {( {1 - \frac{1}{{1{\rm{ + }}x_0^\alpha {\theta _h}{r^{ - \alpha }}}}} )} }  \cdot rdrd\theta } ) \\
&= \exp ( { - {\lambda _h}\frac{{2{\pi ^2}x_0^2{\theta _h}^{\frac{2}{\alpha }}}}{{\alpha \sin ( {\frac{{2\pi }}{\alpha }} )}}}).
\end{aligned}
\end{equation}
\indent Therefore, the success probability of communication for the typical UE is expressed as
\begin{equation}\label{exp8}
{P_{suc}} = \exp ( { - x_0^\alpha {\theta _h}\frac{{N{B_h}}}{{{N_h}{P_h}}} - {\lambda _h}\frac{{2{\pi ^2}x_0^2{\theta _h}^{\frac{2}{\alpha }}}}{{\alpha \sin ( {\frac{{2\pi }}{\alpha }} )}}} ).
\end{equation}
\indent According to (\ref{exp8}), the outage probability in Lemma 1 can be obtained.
\end{proof}
\indent It is revealed from Lemma 1 that the main factor affecting the success probability of communication without spectrum sharing is the interference emitted by HBSs due to spectrum reuse. Therefore, it is necessary to reasonably adjust the density of HBSs, so as to reduce interference and improve network performance.

\subsection{Outage probability with spectrum sharing}
With spectrum sharing, the spectrum band with bandwidth $B_h$ is used not only by HTC network, but also by MTC network. When the typical UE communicates with the HBS, the SINR is
\begin{equation}\label{exp9}
\gamma _h' = \frac{{{P_h}x_0^{ - \alpha }{h_0}}}{{\sum\limits_{i \in {\phi _h}\backslash \{ 0\} } {{P_h}x_i^{ - \alpha }{h_i} + {P_m'}y_0^{ - \alpha }{g_0} + \frac{{N{B_h}}}{{{N_h}}}} }},
\end{equation}
where  $P_m'$ is the transmit power when the MBS uses the shared spectrum to send data packets. $y_0$ is the distance of the typical UE from the MBS. ${{g_0}}$ is the channel gain of the interference link from MBS to the typical UE.

\indent Then, we have Lemma 2.

\begin{lemma}\label{lemma2}
The outage probability of communication for the typical UE in HTC network with spectrum sharing is \cite{7842290}

\begin{equation}\label{exp10}
\begin{aligned}
P_{out}' = &1 - \frac{1}{{1 + x_0^\alpha {\theta _h}\frac{{P_m'}}{{{P_h}}}y_0^{ - \alpha }}}\cdot\\
&\exp (- x_0^\alpha {\theta _h}\frac{{N{B_h}}}{{{N_h}{P_h}}}- {\lambda _h}\frac{{2{\pi ^2}x_0^2{\theta _h}^{\frac{2}{\alpha }}}}{{\alpha \sin ( {\frac{{2\pi }}{\alpha }})}} ).
\end{aligned}
\end{equation}
\end{lemma}
\begin{proof}
Substituting (\ref{exp9}) into (27) in \cite{7842290}, the expression for the success probability of communication is
\begin{equation}
\begin{aligned}
&P_{suc}'\\
&= P( {\gamma _h' > {\theta _h}} )\label{exp11}\\
&= \exp ( { - x_0^\alpha {\theta _h}\frac{{N{B_h}}}{{{N_h}{P_h}}}} ){L_{{I_h}}}( {x_0^\alpha {\theta _h}} )
{E_{{g_0}}}( { - x_0^\alpha {\theta _h}\frac{{{P_m'}}}{{{P_h}}}y_0^{ - \alpha }{g_0}}).
\end{aligned}
\end{equation}
\indent According to (30) in \cite{7842290}, the success probability of communication is
\begin{equation}\label{exp12}
\begin{aligned}
P_{suc}' =& \frac{1}{{1 + x_0^\alpha {\theta _h}\frac{{P_m'}}{{{P_h}}}y_0^{ - \alpha }}}\\
& \cdot\exp ( { - x_0^\alpha {\theta _h}\frac{{N{B_h}}}{{{N_h}{P_h}}} - {\lambda _h}\frac{{2{\pi ^2}x_0^2{\theta _h}^{\frac{2}{\alpha }}}}{{\alpha \sin( {\frac{{2\pi }}{\alpha }})}}}).
\end{aligned}
\end{equation}
\indent According to  (\ref{exp12}), the expression of the outage probability in Lemma 2 can be derived.
\end{proof}
It is revealed from Lemma 2 that with spectrum sharing, the transmit power of MBS has an impact on the performance of HTC network. Therefore, the transmit power of MBS needs to be adjusted to avoid serious interference to the UEs in HTC network, so as to enable the coexistence of HTC network and MTC network.

\indent With spectrum sharing, the increment of outage probability can be expressed as
\begin{equation}\label{exp13}
\delta  = \frac{{P_{out}' - {P_{out}}}}{{{P_{out}}}}.
\end{equation}

\section{Delay and Delay Jitter of MTC Network}

\indent For MTC network, ensuring real-time data transmission is the primary goal. Therefore, delay and jitter are key performance metrics to be analyzed for MTC network.\

\begin{figure}
\centering
\includegraphics[width=0.49\textwidth]{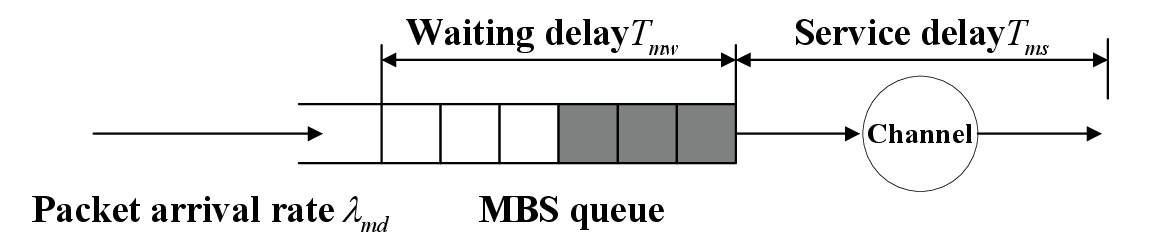}
\caption{Queueing model at MBS in MTC network.}
\vspace{1em}
\rule[5pt]{9cm}{0.05em}
\label{fig_queue model}
\end{figure}

\indent In MTC network, the delay experienced by a data packet includes service delay and waiting delay. Service delay refers to the delay it takes for a data packet to be sent from the transmitter to the receiver \cite{benvenuto2011principles}. Waiting delay is the delay for a data packet to enter the queue until it is transmitted \cite{benvenuto2011principles}. The downlink communication process is shown in Fig. \ref{fig_queue model}. The data packet enters the queue of MBS, and the newly arrived data packet needs to wait until it is transmitted. In this paper, the queueing system at the MBS is modeled as M/G/1 queueing model, where the arrival of data packets at the MBS follows a Poisson process with arrival rate ${\lambda _{md}}$. It is assumed that the service principle of the data packets is First Come First Served (FCFS).

\indent In queueing theory, it is assumed that the arrival rate of packets is $\lambda_{md} $, the service rate is $\mu$, the total number of packets in the queue is $L$, the number of packets in the waiting state is ${L_w}$, and the number of packets being served is ${L_s}$. The average service delay of a packet is ${T_{ms}}$, the average waiting delay is ${T_{mw}}$, and the average sojourn delay is $T_m$. When the system is in a steady state, the following equations can be obtained by using Littles formula \cite{benvenuto2011principles}.

\begin{equation}\label{exp14}
\left\{ \begin{array}{l}
L = {L_w} + {L_s}\\
{T_{mw}} = \frac{{{L_w}}}{{{\lambda _{md}}}}\\
{T_{ms}} = \frac{{{L_s}}}{\mu }\\
{T_m} = {T_{ms}} + {T_{mw}}
\end{array} \right.
\end{equation}

\indent Using queueing theory, this section analyzes the delay and delay jitter of the typical MTC device in MTC network with and without spectrum sharing.
\subsection{ Performance of MTC network without proprietary spectrum}
  When MBS and its MTC devices communicate using shared spectrum, the signals received by the MTC devices include useful signals from MBS, interference signals from HBS and noise. The SINR of the typical MTC device is
\begin{equation}\label{exp15}
{\gamma _{mo}} = \frac{{P_m'y_0^{ - \alpha }{k_0}}}{{\sum\limits_{i \in {\phi _h}} {{P_h}x_i^{ - \alpha }{h_i}} {\rm{ + }}N\frac{{{B_h}}}{{{N_m}}}}},
\end{equation}
where $P_m'$ is the transmit power of MBS using shared spectrum. $y_0$ is the distance between the typical MTC device and MBS. $k_0$ is the channel gain between the typical MTC device and MBS.
The channel capacity can be expressed as
\begin{equation}\label{exp16}
{C_m} = \frac{{{B_h}}}{{{N_m}}}{\log _2}({1{\rm{ + }}\frac{{P_m'y_0^{ - \alpha }{k_0}}}{{\sum\limits_{i \in {\phi _h}} {{P_h}x_i^{ - \alpha }{h_i}} {\rm{ + }}N\frac{{{B_h}}}{{{N_m}}}}}} ).
\end{equation}
Then, the service delay of data packet is
\begin{equation}\label{exp17}
{T_{mso}}{\rm{ = }}\frac{{{U_m}{N_m}}}{{{B_h}{{\log }_2}( {1{\rm{ + }}\frac{{P_m'y_0^{ - \alpha }{k_0}}}{{\sum\limits_{i \in {\phi _h}} {{P_h}x_i^{ - \alpha }{h_i}} {\rm{ + }}N\frac{{{B_h}}}{{{N_m}}}}}})}},
\end{equation}
where ${U_m}$ is the size of the data packet sent by the MBS to the typical MTC device.
\begin{Theo}
In the scenario where MTC network has no proprietary spectrum, the average delay experienced by a data packet in the downlink is
\begin{equation}\label{exp18}
\begin{aligned}
E({T_{mo}}) = &{t_{out}} - \int_0^{{t_{out}}} {{F_{{T_{mso}}}}( t )} dt \\
&+ \frac{{{\lambda _{md}}( {t_{out}^2 - 2\int_0^{{t_{out}}} {t{F_{{T_{mso}}}}( t)} dt} )}}{{2( {1 - {\lambda _{md}}( {{t_{out}} - \int_0^{{t_{out}}} {{F_{{T_{mso}}}}( t )} dt} )})}}.
\end{aligned}
\end{equation}
\indent The delay jitter of MTC network is
\begin{equation}\label{exp19}
\begin{aligned}
{J_{mo}} = &{(\frac{{{\lambda _{md}}(t_{out}^2 - 2\int_0^{{t_{out}}} {t{F_{{T_{mso}}}}(t)} dt)}}{{2(1 - {\lambda _{md}}({t_{out}} - \int_0^{{t_{out}}} {{F_{{T_{mso}}}}(t)} dt))}})^2} \\
&+ \frac{{{\lambda _{md}}(t_{{t_{out}}}^3 - 3\int_0^{{t_{out}}} {{t^2}{F_{{T_{mso}}}}(t)} dt)}}{{2(1 - {\lambda _{md}}({t_{out}} - \int_0^{{t_{out}}} {{F_{{T_{mso}}}}(t)} dt))}},
\end{aligned}
\end{equation}
\indent where
\begin{tiny}
\begin{equation}\label{exp20}
{F_{{T_{mso}}}}(t)=  \exp ({- y_0^\alpha ( {{2^{\frac{{{U_m}{N_m}}}{{{B_h}t}}}} - 1} )\frac{{N{B_h}}}{{P_m'{N_m}}} - {\lambda _h}\frac{{2{\pi ^2}y_0^2{{( {\frac{{{P_h}}}{{P_m'}}( {{2^{\frac{{{U_m}{N_m}}}{{{B_h}t}}}} - 1} )} )}^{\frac{2}{\alpha }}}}}{{\alpha \sin( {\frac{{2\pi }}{\alpha }})}}}.
\end{equation}
\end{tiny}
\end{Theo}

\begin{proof}
The proof of Theorem 1 is in Appendix A.
\end{proof}

\indent When spectrum sharing is adopted, the transmit power of MBS needs to be controlled to reduce the interference to HTC network. In order to ensure that the outage probability of HTC network with spectrum sharing is lower than the corresponding threshold $\varepsilon $, namely
\begin{equation}\label{exp21}
P_{out}' < \varepsilon,
\end{equation}
the transmit power $P_m'$ of MBS in MTC network needs to satisfy the following inequality
\begin{equation} \label{exp22}
\begin{aligned}
P_m' \le &( {\frac{1}{{( {1 - \varepsilon } )\exp ( {x_0^\alpha {\theta _h}\frac{{N {B_h}}}{{{N_h}{P_h}}} + {\lambda _h}\frac{{2{\pi ^2}x_0^2{\theta _h}^{\frac{2}{\alpha }}}}{{\alpha \sin ( {\frac{{2\pi }}{\alpha }} )}}} )}} - 1})\\
&\cdot x_0^{ - \alpha }\frac{{{P_h}}}{{{\theta _h}}}y_0^\alpha.
\end{aligned}
\end{equation}
\indent Meanwhile, the transmit power of MBS cannot exceed ${P_{\max }}$, namely
\begin{equation} \label{exp23}
P_m' \le {P_{\max }}.
\end{equation}
\indent The maximum allowable transmit power of MBS when MTC network shares spectrum with HTC network can be obtained by (\ref{exp22}) and (\ref{exp23}), and the maximum delay and jitter of MTC network when the interference of MBS to HTC network is within a certain range can be obtained by substituting the value of transmit power into (\ref{exp18}) and (\ref{exp19}).

\subsection{ Performance of MTC network with proprietary spectrum}
\subsubsection{Delay and jitter without spectrum sharing}
 When the typical MTC device communicates with the MBS, the received signals include the desired signal from the MBS and noise. The Signal to Noise Ratio (SNR) of the typical MTC device can be expressed as
\begin{equation}\label{exp24}
\gamma _m= \frac{{{P_m}{N_m}y_0^{ - \alpha }{{\rm{g}}_0}}}{{N{B_m}}}.
\end{equation}
\indent When communicating with the MBS, the typical MTC device uses the spectrum with bandwidth $\frac{{{B_m}}}{{{N_m}}}$. The expression of the service delay is
\begin{equation}\label{exp25}
{T_{ms}}{\rm{ = }}\frac{{{U_m}{N_m}}}{{{B_m}{{\log }_2}( {1 + \gamma _m})}}.
\end{equation}

\indent Substituting (\ref{exp24}) into (\ref{exp25}), the service delay of the typical MTC device is
\begin{equation}\label{exp26}
{T_{ms}}{\rm{ = }}\frac{{{U_m}{N_m}}}{{{B_m}{{\log }_2}( {1{\rm{ + }}\frac{{{P_m}{N_m}y_0^{ - \alpha }{{\rm{g}}_0}}}{{N{B_m}}}} )}},
\end{equation}
where ${B_m}$ is the bandwidth of proprietary spectrum of MBS. ${P_m}$ is the transmit power when the MBS uses the proprietary spectrum. ${y_0}$ is the distance from the typical MTC device to the MBS. ${g_0}$ is the {channel gain when the typical MTC device communicates with the MBS}.

\indent Then, we have Theorem 2.
\begin{Theo}
Without spectrum sharing, the average delay experienced by a data packet in the downlink is
\begin{equation}\label{exp27}
\begin{aligned}
{T_m} &={t_{out}} - \int_0^{{t_{out}}} {{F_{{T_{ms}}}}( t )} dt\\
&{\rm{ + }}\frac{{{\lambda _{md}}( {t_{out}^2 - 2\int_0^{{t_{out}}} {t{F_{{T_{ms}}}}( t )} dt} )}}{{2( {1 - {\lambda _{md}}( {{t_{out}} - \int_0^{{t_{out}}} {{F_{{T_{ms}}}}( t )} dt} )} )}}.
\end{aligned}
\end{equation}
\indent Delay jitter of MTC network is
\begin{equation}\label{exp28}
\begin{aligned}
{J_m} = &t_{out}^2 - 2\int_0^{{t_{out}}} {t{F_{{T_{ms}}}}( t )} dt - {( {{t_{out}} - \int_0^{{t_{out}}} {{F_{{T_{ms}}}}( t )} dt} )^2}\\
 &+ {( {\frac{{{\lambda _{md}}( {t_{out}^2 - 2\int_0^{{t_{out}}} {t{F_{{T_{ms}}}}( t )} dt} )}}{{2( {1 - {\lambda _{md}}( {{t_{out}} - \int_0^{{t_{out}}} {{F_{{T_{ms}}}}( t )} dt} )} )}}} )^2}\\
 &+ \frac{{{\lambda _{md}}( {t_{{t_{out}}}^3 - 3\int_0^{{t_{out}}} {{t^2}{F_{{T_{ms}}}}( t )} dt} )}}{{2( {1 - {\lambda _{md}}( {{t_{out}} - \int_0^{{t_{out}}} {{F_{{T_{ms}}}}( t )} dt} )} )}},
\end{aligned}
\end{equation}
\indent where
\begin{equation}\label{exp29}
{F_{{T_{ms}}}}( t ){\rm{ = }}\exp ( { - y_0^\alpha ( {{2^{\frac{{{U_m}{N_m}}}{{{B_m}t}}}} - 1} )\frac{{N{B_m}}}{{{P_m}{N_m}}}} ).
\end{equation}
\end{Theo}

\begin{proof}
The proof of Theorem 2 is in Appendix B.
\end{proof}

\subsubsection{Delay and jitter with spectrum sharing}
When the MBS communicates with its MTC device using proprietary spectrum, the SNR of the typical MTC device is
\begin{equation}\label{exp30}
\gamma _{m1}' = \frac{{{P_m}y_0^{ - \alpha }{g_0}}}{{N\frac{{{B_m}}}{{{N_m}}}}} = \frac{{{P_m}{N_m}y_0^{ - \alpha }{g_0}}}{{N{B_m}}}.
\end{equation}
\indent When the MBS communicates with the typical MTC device using the shared spectrum, the SINR of the typical MTC device can be expressed as
\begin{equation}\label{exp31}
\gamma _{m2}' = \frac{{{P_m'}y_0^{ - \alpha }{k_0}}}{{\sum\limits_{i \in {\phi _h}} {{P_h}x_i^{ - \alpha }{h_i}} {\rm{ + }}\frac{{N{B_h}}}{{{N_m}}}}},
\end{equation}
where ${k_0}$ is the channel gain when the typical MTC device communicates with the MBS. ${x_i}$ is the distance between the typical MTC device and the $i$th HBS. ${h_i}$ is the channel gain of the interference link from the MBS to the typical MTC device.

\indent The channel capacity when the typical MTC device communicates with the MBS can be expressed as
\begin{equation}\label{exp32}
\begin{aligned}
C_{_{ms}}' = & \frac{{{B_h}}}{{{N_m}}}{\log _2}( {1{\rm{ + }}\frac{{{P_m'}y_0^{ - \alpha }{k_0}}}{{\sum\limits_{i \in {\phi _h}} {{P_h}x_i^{ - \alpha }{h_i}} + \frac{{N{B_h}}}{{{N_m}}}}}} )\\
 &+ \frac{{{B_m}}}{{{N_m}}}{\log _2}( {1{\rm{ + }}\frac{{{P_m}{N_m}y_0^{ - \alpha }{g_0}}}{{N{B_m}}}} ).
\end{aligned}
\end{equation}
\indent The service delay is
\begin{small}
\begin{equation}\label{exp33}
T_{ms}'{\rm{ = }}\frac{{{U_m}}}{{\frac{{{B_h}}}{{{N_m}}}{{\log }_2}( {1{\rm{ + }}\frac{{{P_m'}y_0^{ - \alpha }{k_0}}}{{\sum\limits_{i \in {\phi _h}} {{P_h}x_i^{ - \alpha }{h_i}}  + \frac{{N{B_h}}}{{{N_m}}}}}}) + \frac{{{B_m}}}{{{N_m}}}{{\log }_2}( {1{\rm{ + }}\frac{{{P_m}{N_m}y_0^{ - \alpha }{g_0}}}{{N{B_m}}}})}}.
\end{equation}
\end{small}
Then, we have Theorem 3.
\begin{spacing}{1.2}
\begin{Theo}
With spectrum sharing, the average delay experienced by a data packet in the downlink is
\begin{footnotesize}
\begin{equation}\label{exp34}
\begin{aligned}
E( {T_m'} )= &\int_0^{{t_{out}}} { {\int_0^{\frac{{{U_m}{N_m}}}{t}} {\;{F_1}( \tau  ) \cdot {f_2}( {\frac{{{U_m}{N_m}}}{t} - \tau } )} d\tau } dt}  \\
&+ \frac{{{\lambda _{md}}( {\int_0^{{t_{out}}} { {\int_0^{\frac{{{U_m}{N_m}}}{t}} {\;t{F_1}( \tau  ) \cdot {f_2}( {\frac{{{U_m}{N_m}}}{t} - \tau })} d\tau } dt} } )}}{{1 - {\lambda _{md}}( {\int_0^{{t_{out}}} { {\int_0^{\frac{{{U_m}{N_m}}}{t}} {\;{F_1}( \tau ) \cdot {f_2}( {\frac{{{U_m}{N_m}}}{t} - \tau })} d\tau } dt} } )}}.
\end{aligned}
\end{equation}
\end{footnotesize}
\indent Delay jitter of MTC network is
\begin{small}
\begin{equation}\label{exp35}
\begin{aligned}
J_m'=& {\rm{2}}\int_0^{{t_{out}}} { {\int_0^{\frac{{{U_m}{N_m}}}{t}} {\;t{F_1}( \tau ) \cdot {f_2}( {\frac{{{U_m}{N_m}}}{t} - \tau } )} d\tau } dt}\\
 &- {( {\int_0^{{t_{out}}} { {\int_0^{\frac{{{U_m}{N_m}}}{t}} {\;{F_1}( \tau ) \cdot {f_2}( {\frac{{{U_m}{N_m}}}{t} - \tau } )} d\tau } dt} } )^2}\\
 &+ {( {\frac{{{\lambda _{md}}( {\int_0^{{t_{out}}} { {\int_0^{\frac{{{U_m}{N_m}}}{t}} {\;t{F_1}( \tau  ) \cdot {f_2}( {\frac{{{U_m}{N_m}}}{t} - \tau } )} d\tau } dt} } )}}{{1 - {\lambda _{md}}( {\int_0^{{t_{out}}} { {\int_0^{\frac{{{U_m}{N_m}}}{t}} {\;{F_1}( \tau  ) \cdot {f_2}( {\frac{{{U_m}{N_m}}}{t} - \tau } )} d\tau } dt} } )}}} )^2} \\
 &+ \frac{{{\lambda _{md}}( {\int_0^{{t_{out}}} { {\int_0^{\frac{{{U_m}{N_m}}}{t}} {\;{t^2}{F_1}( \tau  ) \cdot {f_2}( {\frac{{{U_m}{N_m}}}{t} - \tau } )} d\tau } dt} } )}}{{1 - {\lambda _{md}}( {\int_0^{{t_{out}}} { {\int_0^{\frac{{{U_m}{N_m}}}{t}} {\;{F_1}( \tau  ) \cdot {f_2}( {\frac{{{U_m}{N_m}}}{t} - \tau } )} d\tau } dt} } )}},
\end{aligned}
\end{equation}
\end{small}
\indent where
\begin{equation}\label{exp36}
\begin{aligned}
{F_1}( \tau  ) =& 1 - \exp ( - y_0^\alpha ( {{2^{\frac{\tau }{{{B_h}}}}} - 1} )\frac{{N{B_h}}}{{{P_m'}{N_m}}} \\
&- {\lambda _h}\frac{{2{\pi ^2}y_0^2{{( {\frac{{{P_h}}}{{{P_m'}}}( {{2^{\frac{\tau }{{{B_h}}}}} - 1} )} )}^{\frac{2}{\alpha }}}}}{{\alpha \sin \left( {\frac{{2\pi }}{\alpha }} \right)}}),
\end{aligned}
\end{equation}
\begin{small}
\begin{equation}\label{exp37}
{f_2}( \tau  ) = {\rm{exp}}( { - y_0^\alpha ( {{2^{\frac{\tau }{{{B_m}}}}} - 1} )\frac{{N{B_m}}}{{{P_m}{N_m}}}} )( {\frac{{{2^{\frac{\tau }{{{B_m}}}}}N{{\log }_e}( 2 )y_0^\alpha }}{{{P_m}{N_m}}}}).
\end{equation}
\end {small}
\end{Theo}
\end{spacing}

\begin{proof}
The proof of Theorem 3 is in Appendix C.
\end{proof}

\section{Analysis and Guiding Significance of Simulation Results}
\indent In this section, the performance of HTC network and MTC network is verified by simulation. For HTC network, the outage probability of communication with and without spectrum sharing are compared. For MTC network, the delay and jitter of communication with and without spectrum sharing are compared.
\subsection{Simulation Results and Analysis}
Monte Carlo simulation is applied to verify the correctness of the theoretical analysis.
Due to the randomness of the distribution of BSs and users in practical applications,
stochastic geometry is used to analyze the network performance by treating the BSs and users as random points.
First, according to the property of homogeneous PPP,
a specific number of points are scattered to simulate the deployment of HBSs,
UE and MTC devices.
Then, the distance between the typical UE and HBS as well as
the distance between the typical MTC device and MBS are calculated.
The channel gains are randomly generated according to the characteristics of the exponential distribution.
Finally, the outage probability of HTC network and the delay performance of MTC network are simulated, and the theoretical results are also obtained.
The simulation parameters are determined according to the 3GPP standard.
The values of the transmit power of HBS and MBS refer to the transmit power of the BS in \cite{3gpp.38.901}.
The simulation parameters,
such as the size of data packets and the distribution density of MTC devices, are selected based on the service performance requirements proposed by
3GPP TS 22.104 \cite{3gpp.22.104}.
The setting of simulation parameters is listed in Table 2.

\begin{table}[h]
 \caption{\label{sys_para}The Setting of the Simulation Parameters}
 \begin{center}
 \begin{tabular}{l l}
 \hline
 \hline

    {Notation} & {Value} \\

  \hline
  $ x_0 $ & $10$ m \\
  $ y_0 $ & $10$ m \\
  $ B_h $ & $20$ MHz\\
  $ B_m $ & $100$ MHz \\
  $ N $ & ${10^{ - 10}}$ {\rm{W}}/Hz\\
  $ \alpha $ & 4 \\
  $ U_m $ & $40$ byte \\
  $ t_{out} $ & $0.01$ s\\
  $ N_h $ & 1000 \\
  $ \theta_h $ &  0.01\\
  $ P_{\max } $ & $24$ dBm\\
  \hline
  \hline
 \end{tabular}
 \end{center}
\end{table}

\begin{figure}
\centering
\includegraphics[width=0.49\textwidth]{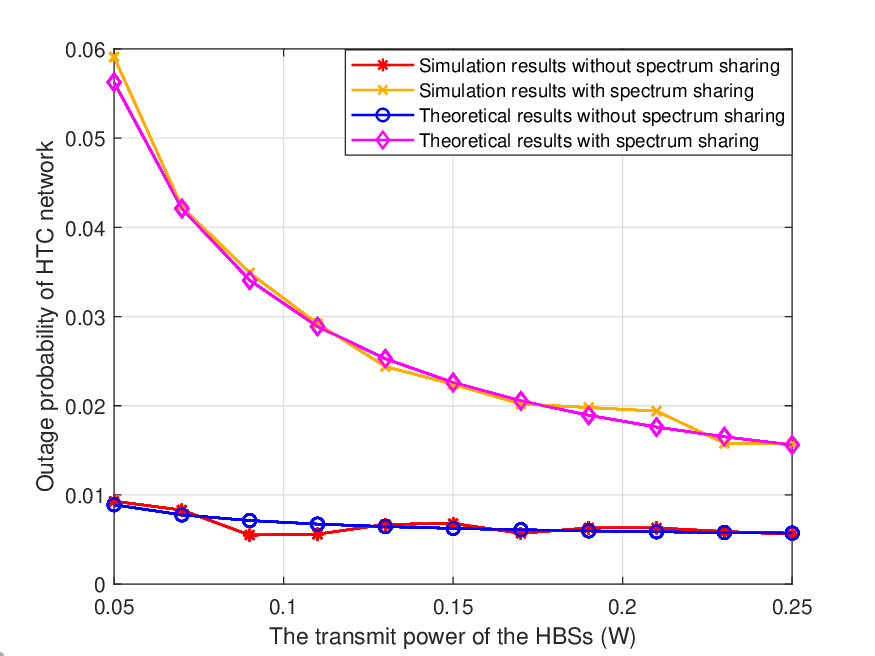}
\caption{The relationship between the outage probability and the transmit power of HBSs.}
\label{fig_Outage1}
\end{figure}
\indent Fig. \ref{fig_Outage1} shows the impact of the transmit power of HBSs on the outage probability of communication. It is revealed that the outage probability is hardly changed without spectrum sharing. The increased strength of signal and interference is the same, which results in the stability of the outage probability. With spectrum sharing, the outage probability of HTC network gradually decreases. This is due to the fact that the interference from MBS to HTC network is reduced with the increase of the transmit power of HBSs.

\begin{figure}
\centering
\includegraphics[width=0.49\textwidth]{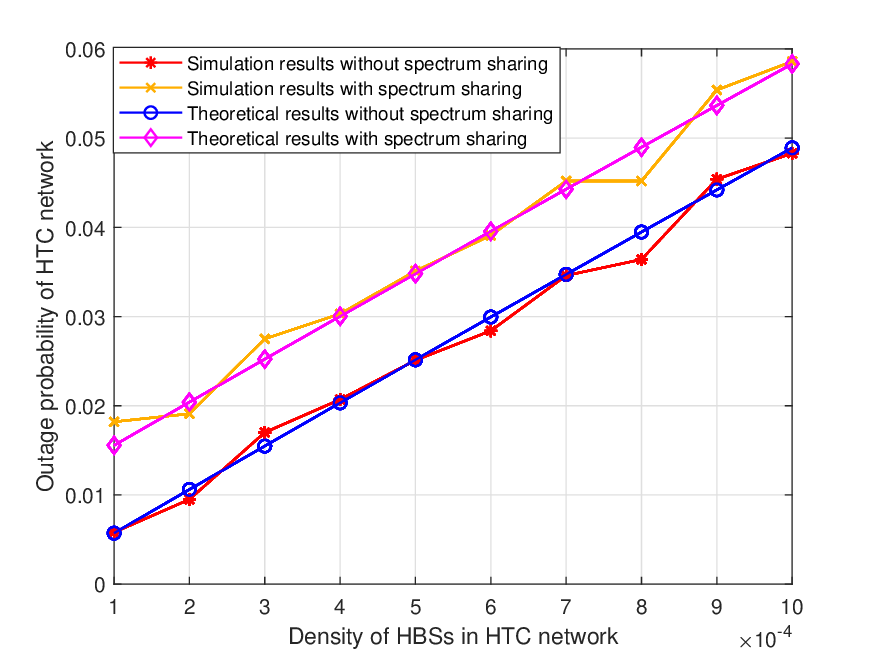}
\caption{The relationship between the outage probability and the distribution density of HBSs.}
\label{fig_Outage2}
\end{figure}
\indent Fig. \ref{fig_Outage2} shows the influence of the distribution density of HBSs on the outage probability of communication when the transmit power of HBSs is 24 dBm. As illustrated in Fig. \ref{fig_Outage2}, the outage probability will increase with the increase of the distribution density of HBSs regardless of whether spectrum sharing is adopted. This is due to the fact that the number of HBSs in the same area increases with the increase of the distribution density of HBSs, resulting in an increase in the strength of the interference received by the typical UE and a decrease in the SINR of the typical UE. Therefore, when deploying HBSs in HTC network, the distribution density needs to be reasonably designed to ensure the communication performance of the UEs.

\begin{figure}
\centering
\includegraphics[width=0.49\textwidth]{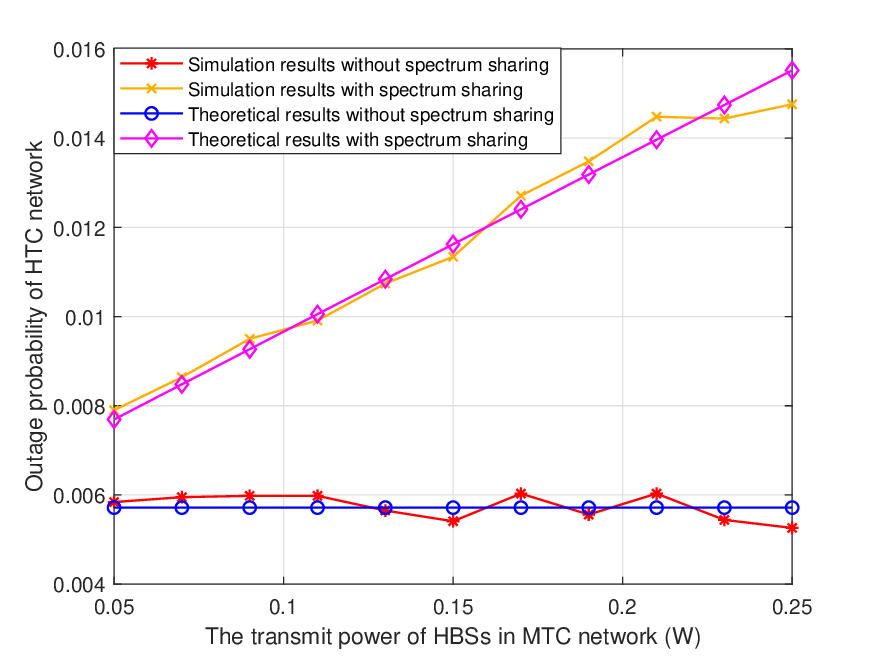}
\caption{The relationship between the outage probability and the transmit power of MBS in MTC network.}
\label{fig_Outage3}
\end{figure}
\indent The influence of the transmit power of MBS on the outage probability of HTC network with and without spectrum sharing is shown in Fig. \ref{fig_Outage3}. MTC network has only one MBS. In the simulation, the channel gain of MBS is generated randomly. As illustrated in Fig. \ref{fig_Outage3}, without spectrum sharing, the communication between the typical UE and its associated HBS is not interfered with MTC network since the two networks are independent. Hence, the outage probability of HTC network is basically stable. With spectrum sharing, the outage probability of HTC network increases with the increase of the transmit power of MBS, which is mainly caused by the increased interference of MBS to HTC network. Therefore, as long as the transmit power of  MBS is reasonably controlled, the communication of HTC network will not be affected severely. The HTC network and the MTC network share the spectrum, thereby improving the spectrum utilization and the capacity of MTC network simultaneously.

\begin{figure}
\centering
\includegraphics[width=0.49\textwidth]{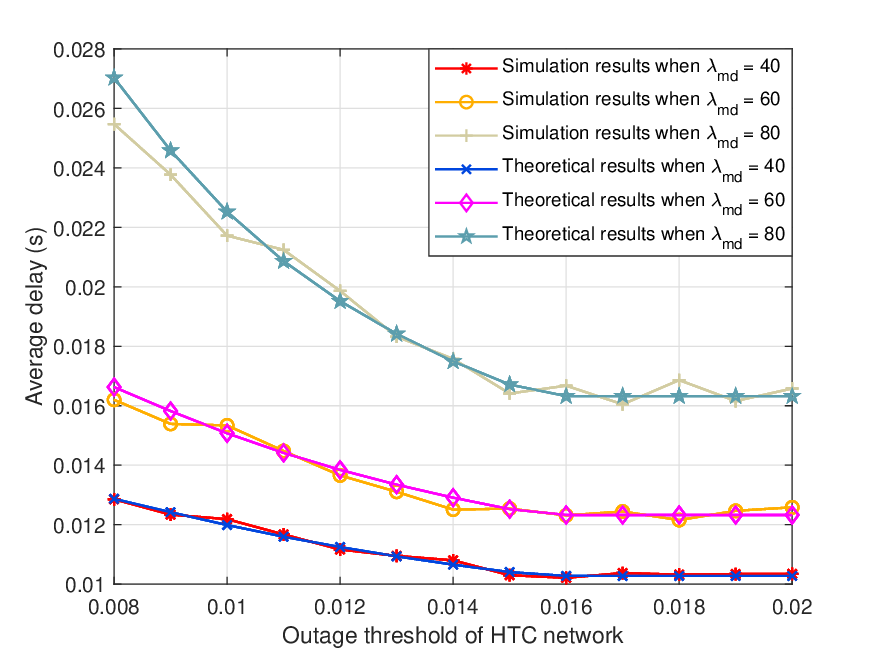}
\caption{The relationship between average delay and outage threshold of HTC network when MTC network has no proprietary spectrum.}
\label{fig_delay_no}
\end{figure}
\indent Fig. \ref{fig_delay_no} shows the delay performance when sharing spectrum with HTC network
in the scenario that MTC network has no proprietary spectrum. According to the comparison of the curves
in Fig. \ref{fig_delay_no}, it is revealed that with the increase of the packet arrival rate, the average delay of a data packet gradually increases.
The reason for this trend is that the number of data packets
arriving at the MBS increases per unit time,
and more data packets need to wait for a long time to be served.
In addition, Fig. \ref{fig_delay_no} shows that with the increase of the outage threshold, the average delay experienced by a data packet first decreases,
and then tends to be stable.
This is due to the increase in outage threshold, which increases the maximum allowable transmit power of MBS,
thereby improving the service capability of
MBS and reducing the average delay. However, when the transmit power of MBS reaches the maximum
allowable transmit power, it will maintain this
power to transmit data packets.

\begin{figure}
\centering
\includegraphics[width=0.49\textwidth]{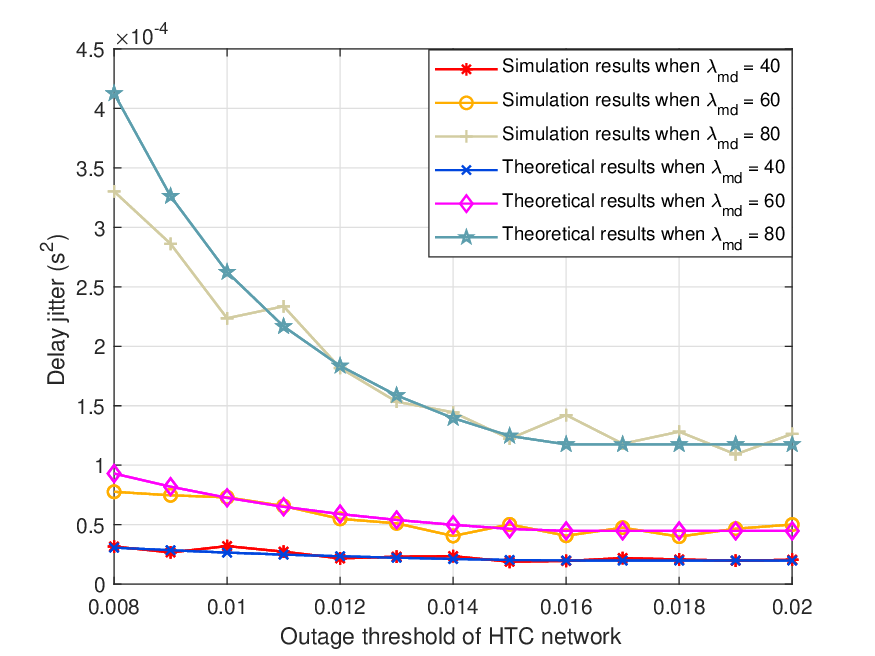}
\caption{The relationship between delay jitter and outage threshold of HTC network when MTC network has no proprietary spectrum.}
\label{fig_jitter_no}
\end{figure}
\indent Fig. \ref{fig_jitter_no} shows the relationship between
the delay jitter and the communication outage threshold.
It is revealed that the change trend of delay jitter is consistent
with that of average delay.
Fig. \ref{fig_jitter_no} shows that with the increase of
the arrival rate of data packets,
the delay jitter gradually increases.
The reason for this trend is that as the arrival rate of data packets
increases, the uncertainty of waiting delay before being served increases,
thus leading to increased jitter. It is revealed from
Fig \ref{fig_jitter_no} that the delay
jitter decreases with the increase of the outage threshold and finally tends to be stable.
The reason is that the increase of the outage threshold increases
the maximum allowable transmit power of MBS, so that the service delay decreases,
thus reducing the variance of the service delay.
The reason of the stabilized trend of delay jitter is that the transmit
power of MBS reaches the maximum allowable transmit power.

\begin{figure}
\centering
\includegraphics[width=0.49\textwidth]{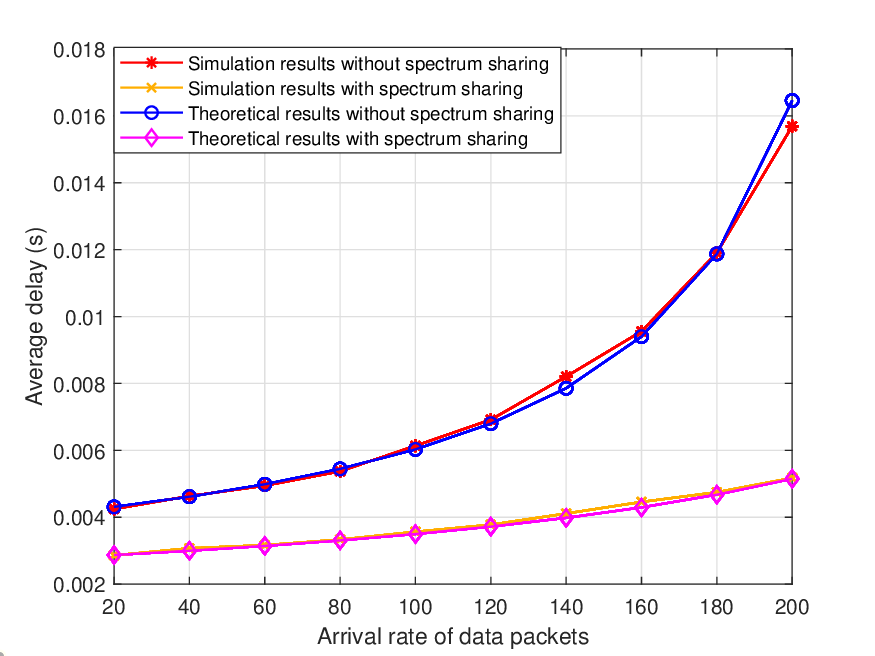}
\caption{The influence of packet arrival rate on average delay with and without spectrum sharing.}
\label{fig_Average1}
\end{figure}
\indent
The change of the average delay of a data packet with and without spectrum sharing is shown in Fig. \ref{fig_Average1}. As the arrival rate of data packets increases, the average delay of data packets increases gradually. The reason is that the increase in the number of data packets arriving in the queue per unit time leads to an increase in the waiting delay of packets. It can also be observed from the comparison that the average delay of a data packet is reduced with spectrum sharing. This is due to the fact that spectrum sharing improves the communication capacity of MTC network, so that the waiting delay of a data packet is reduced, thereby reducing the average delay.
\begin{figure}
\centering
\includegraphics[width=0.49\textwidth]{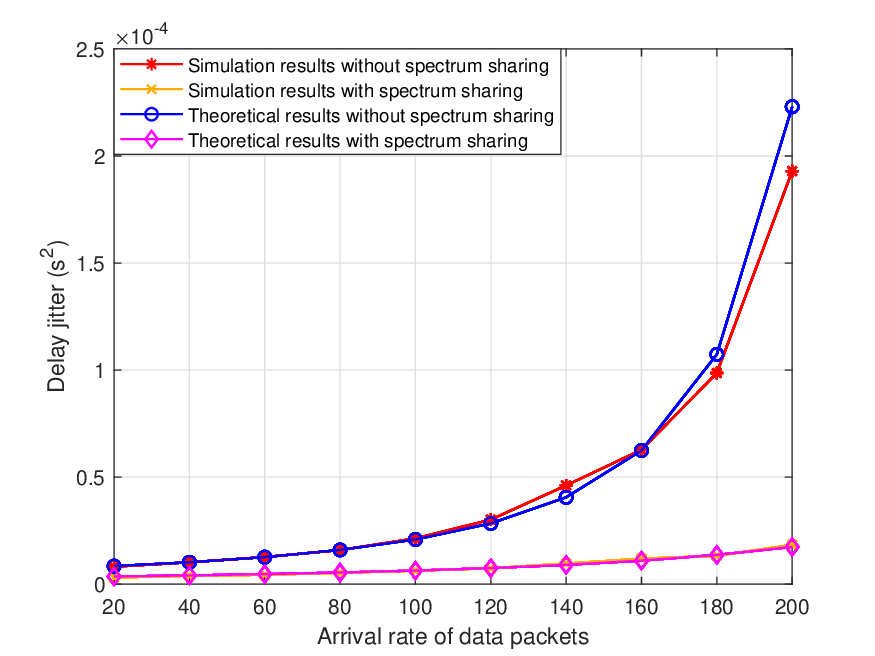}
\caption{The influence of packet arrival rate on delay jitter with and without spectrum sharing.}
\label{fig_Jitter1}
\end{figure}

\indent Fig. \ref{fig_Jitter1} depicts the variation of
delay jitter when the arrival rate of data packets increases gradually.
It can be observed from the curves that the spectrum sharing
strategy can significantly reduce the delay jitter.
Besides, as the arrival rate of data packets increases,
the delay jitter of data transmission increases continuously.
The main reason is that the number of data packets arriving in the
queue per unit time increases as the arrival rate of data packets increases,
and the probability that data packets need to wait increases.
As a result, the variance of the waiting delay is increasing.
Therefore, in practical scenarios, the arrival rate of data packets
should be reasonably designed to avoid queue overflow.
\begin{figure}
\centering
\includegraphics[width=0.49\textwidth]{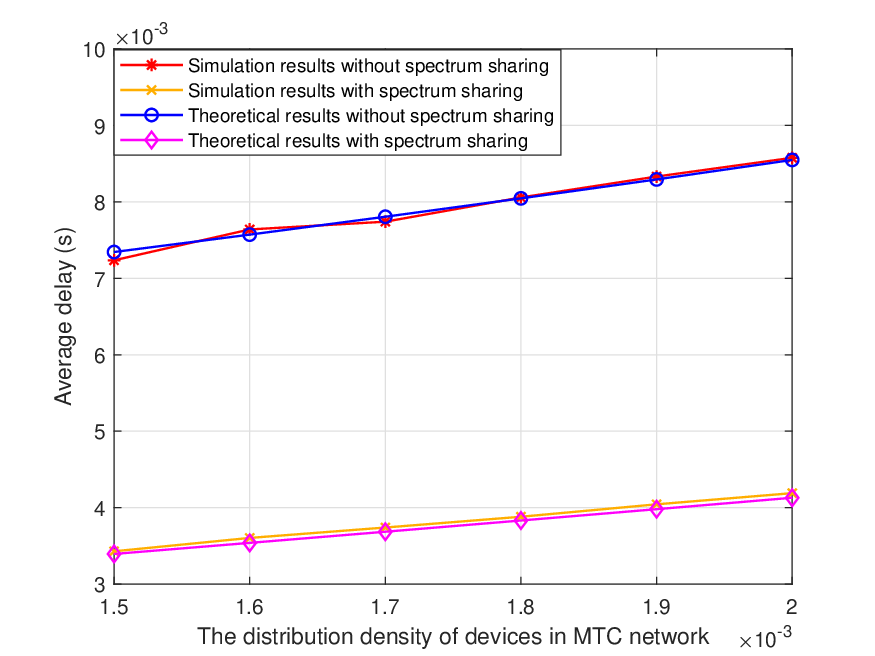}
\caption{The relationship between average delay and distribution density of MTC devices.}
\label{fig_Average2}
\end{figure}

\indent Fig. \ref{fig_Average2} shows the relationship between the average delay and the distribution density of MTC devices. Due to the randomness of the channel gain generated by MBS, the number of iterations is increased in Fig. \ref{fig_Average2} to reduce the gap between theoretical results and simulation results. It is observed that the average delay decreases significantly with spectrum sharing. However, as the distribution density of MTC devices increases, the average delay of a data packet increases gradually. This is due to the fact that the increase in the distribution density of MTC devices leads to the reduction of spectrum available to MTC devices, so that the service delay and waiting delay of a data packet will increase.
\begin{figure}
\centering
\includegraphics[width=0.49\textwidth]{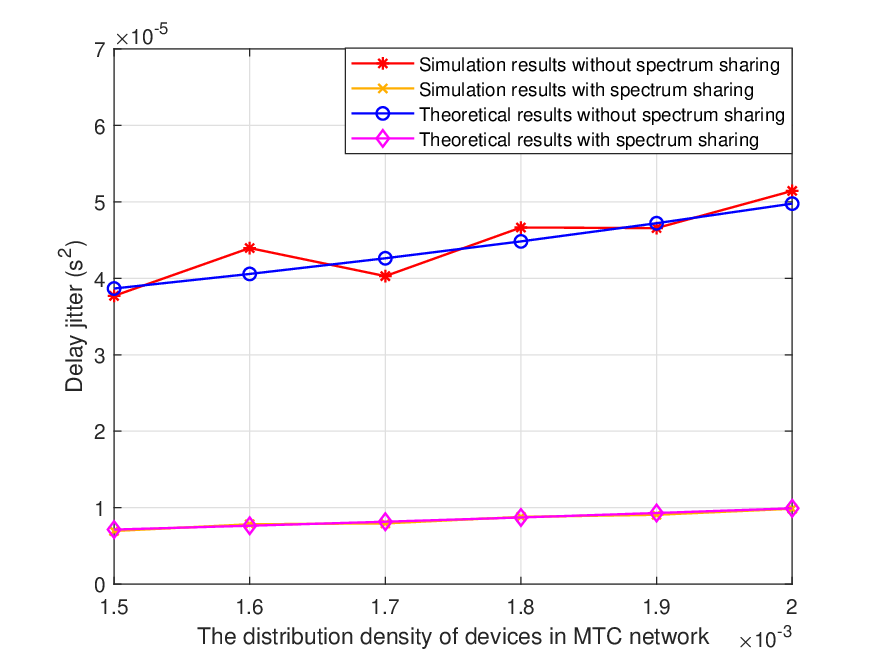}
\caption{The relationship between delay jitter and distribution density of MTC devices.}
\label{fig_Jitter2}
\end{figure}

\indent Fig. \ref{fig_Jitter2} depicts the variation of delay jitter when the distribution density of MTC devices increases. It is observed that the delay jitter is increasing with the increase of the distribution density of MTC devices, whether spectrum sharing is adopted or not. The reason is that as the density of MTC devices increases, the available spectrum obtained by each MTC device decreases, which results in an increase in service delay. The waiting delay and its fluctuation are increasing, which leads to an increase in the delay jitter. Therefore, it is necessary to reasonably deploy the MTC devices to ensure the low delay jitter of MTC network.

\begin{figure}
\centering
\includegraphics[width=0.49\textwidth]{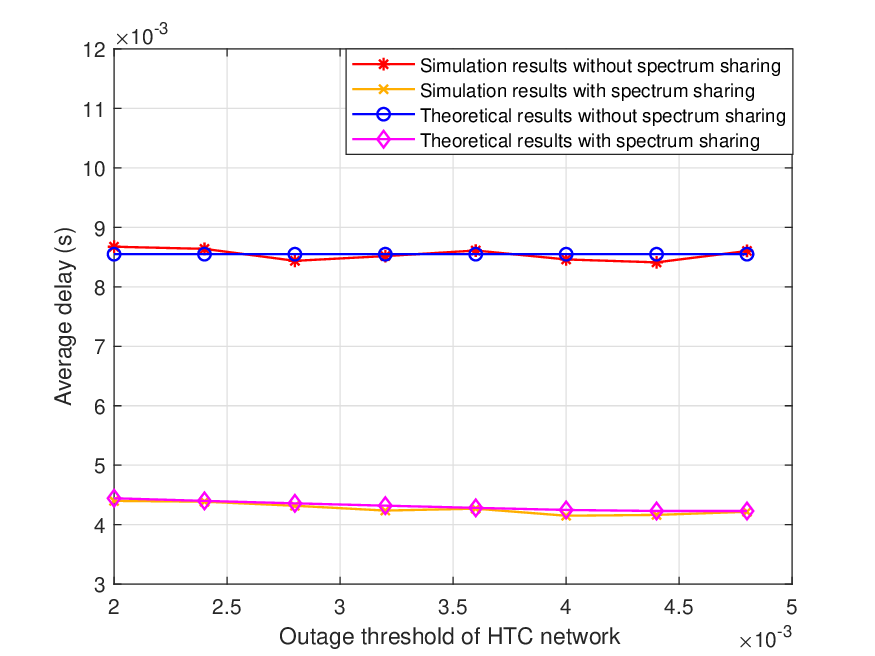}
\caption{The relationship between average delay and outage threshold of HTC network when MTC network has proprietary spectrum.}
\label{fig_Average3}
\end{figure}
\indent The impact of the outage threshold of HTC network on the average delay of MTC network is depicted in Fig. \ref{fig_Average3}. It can be observed that with the increase of the outage threshold, the average delay of MTC network gradually decreases. The reason is that the increase of the outage threshold of HTC network increases the transmit power of MBS with spectrum sharing, thus reducing the service delay of data packet. However, with the increase of the outage threshold, the MBS reaches the maximum allowable transmit power and the average delay of MTC network reaches a stable state.
\begin{figure}
\centering
\includegraphics[width=0.49\textwidth]{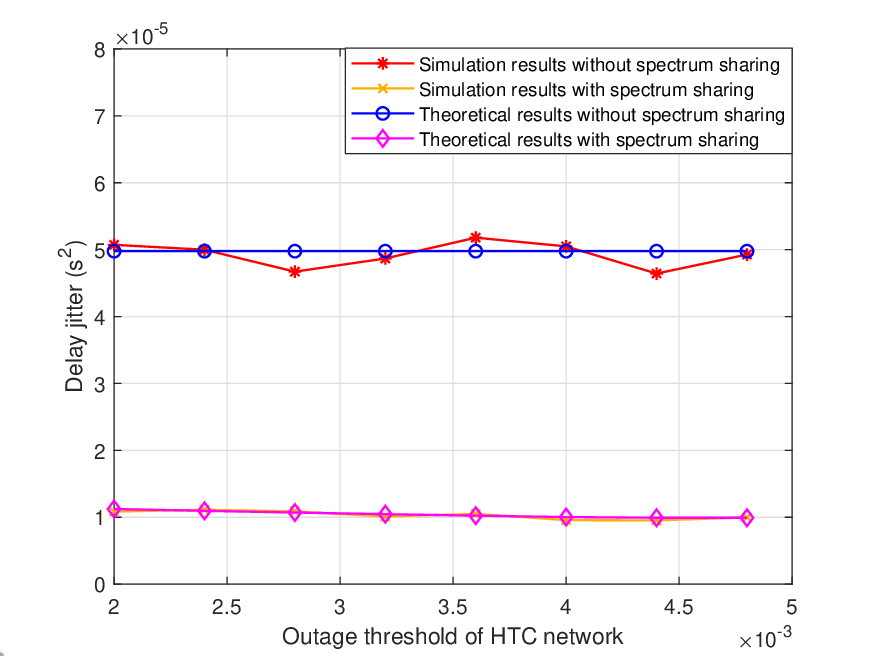}
\caption{The relationship between  delay jitter and outage threshold of HTC network when MTC network has proprietary spectrum.}
\label{fig_Jitter3}
\end{figure}

\indent Fig. \ref{fig_Jitter3} depicts the variation of delay jitter. It can be observed that with the increase of the outage threshold, the delay jitter gradually decreases until it becomes stable. The reason for the downward trend of delay jitter is that with the increase of HTC network's tolerance to communication outage, the transmit power of MBS with spectrum sharing increases. And the service delay and waiting delay of the data packets are reduced, therefore reducing the delay jitter. The delay jitter finally remains stable because the transmit power of MBS has reached the set maximum value.

\subsection{Design Guideline in Industrial Applications}
\indent The simulation results show that MTC network can coexist with HTC network in the same frequency band when it has proprietary spectrum. In practical industrial applications, it is necessary to reasonably control the transmit power of MBS to keep its interference to HTC network within a certain range. Meanwhile, other technologies, such as beamforming, can also be applied to reduce the interference to HTC network. In addition, the density of MBS devices and the arrival rate of data packets will also affect the delay. If the arrival rate of data packets is too large, the limited service capacity will lead to an increase in the waiting delay and even result in network collapse. Therefore, for MTC network, it can also effectively reduce the communication delay and jitter by controlling the density of MBS devices and the arrival rate of data packets, thus improving the performance of MTC network.

\section{Conclusion}
In this paper, the performance of HTC network and MTC network is analyzed.
By leveraging mathematical tools such as stochastic geometry and
queueing theory, this paper obtains expressions for the outage
probability of HTC network, the average delay and jitter of MTC network
with and without spectrum sharing.
The simulation results show that spectrum sharing can reduce the
delay and jitter of MTC network on the premise of controlling the
transmit power of MBS to reduce the interference to HTC network.
Besides, this paper also analyzes the impact of transmit power,
node density and other factors on the
delay performance of MTC network. Simulation results show that
with the continuous increase of the arrival rate of data packets and
the distribution density of MTC devices,
the average delay and jitter of data transmission increase
accordingly. Therefore, it is necessary to
reasonably control the arrival rate of data packets and the
distribution density of MTC devices to
ensure low delay and low jitter in MTC network.
Spectrum sharing has been applied in 4G and 5G,
evolving into LTE-Unlicensed (LTE-U) and 5G New Radio in Unlicensed
Spectrum (5G NR-U) respectively,
so that spectrum sharing is promising to realize the
delay deterministic wireless network.

\bibliographystyle{IEEEtran}
\bibliography{reference}

\begin{thebibliography}{10}
\providecommand{\url}[1]{#1}
\csname url@samestyle\endcsname
\providecommand{\newblock}{\relax}
\providecommand{\bibinfo}[2]{#2}
\providecommand{\BIBentrySTDinterwordspacing}{\spaceskip=0pt\relax}
\providecommand{\BIBentryALTinterwordstretchfactor}{4}
\providecommand{\BIBentryALTinterwordspacing}{\spaceskip=\fontdimen2\font plus
\BIBentryALTinterwordstretchfactor\fontdimen3\font minus
  \fontdimen4\font\relax}
\providecommand{\BIBforeignlanguage}[2]{{%
\expandafter\ifx\csname l@#1\endcsname\relax
\typeout{** WARNING: IEEEtran.bst: No hyphenation pattern has been}%
\typeout{** loaded for the language `#1'. Using the pattern for}%
\typeout{** the default language instead.}%
\else
\language=\csname l@#1\endcsname
\fi
#2}}
\providecommand{\BIBdecl}{\relax}
\BIBdecl

\bibitem{7123559}
K.~Suto, H.~Nishiyama, N.~Kato, and C.-W. Huang, ``An energy-efficient and
  delay-aware wireless computing system for industrial wireless sensor
  networks,'' \emph{IEEE Access}, vol.~3, pp. 1026--1035, 2015.

\bibitem{5764568}
M.~Maadani, S.~A. Motamedi, and M.~M. Noshari, ``Delay analysis and improvement
  of ieee 802.11e-based soft-real-time wireless industrial networks: Using an
  open-loop spatial multiplexing scheme,'' in \emph{2011 International
  Symposium on Computer Networks and Distributed Systems (CNDS)}, 2011, pp.
  17--22.

\bibitem{9606568}
Z.~Jin, C.~Zhang, Y.~Jin, L.~Zhang, and J.~Su, ``A resource allocation scheme
  for joint optimizing energy-consumption and delay in collaborative edge
  computing-based industrial iot,'' \emph{IEEE Transactions on Industrial
  Informatics}, pp. 1--1, 2021.

\bibitem{9560065}
V.~D. Tuong, W.~Noh, and S.~Cho, ``Delay minimization for noma-enabled mobile
  edge computing in industrial internet of things,'' \emph{IEEE Transactions on
  Industrial Informatics}, pp. 1--1, 2021.

\bibitem{9725256}
H.~Yang, J.~Zhao, K.-Y. Lam, Z.~Xiong, Q.~Wu, and L.~Xiao, ``Distributed deep
  reinforcement learning-based spectrum and power allocation for heterogeneous
  networks,'' \emph{IEEE Transactions on Wireless Communications}, vol.~21,
  no.~9, pp. 6935--6948, 2022.

\bibitem{10018944}
T.~Zhang, K.-Y. Lam, J.~Zhao, F.~Li, H.~Han, and N.~Jamil, ``Enhancing
  federated learning with spectrum allocation optimization and device
  selection,'' \emph{IEEE/ACM Transactions on Networking}, pp. 1--16, 2023.

\bibitem{8489935}
G.~Chen, X.~Cao, L.~Liu, C.~Sun, and Y.~Cheng, ``Joint scheduling and channel
  allocation for end-to-end delay minimization in industrial wirelesshart
  networks,'' \emph{IEEE Internet of Things Journal}, vol.~6, no.~2, pp.
  2829--2842, 2019.

\bibitem{8964351}
H.~Wang, S.~Tan, Y.~Zhu, and M.~Li, ``Deterministic scheduling with
  optimization of average transmission delays in industrial wireless sensor
  networks,'' \emph{IEEE Access}, vol.~8, pp. 18\,852--18\,862, 2020.

\bibitem{8882510}
Q.~Li, N.~Zhang, M.~Cheffena, and X.~Shen, ``Channel-based optimal back-off
  delay control in delay-constrained industrial wsns,'' \emph{IEEE Transactions
  on Wireless Communications}, vol.~19, no.~1, pp. 696--711, 2020.

\bibitem{8825819}
M.~Zhang, J.~Chen, S.~He, L.~Yang, X.~Gong, and J.~Zhang, ``Privacy-preserving
  database assisted spectrum access for industrial internet of things: A
  distributed learning approach,'' \emph{IEEE Transactions on Industrial
  Electronics}, vol.~67, no.~8, pp. 7094--7103, 2020.

\bibitem{6883319}
F.~Lin, C.~Chen, L.~Li, H.~Xu, and X.~Guan, ``A novel spectrum sharing scheme
  for industrial cognitive radio networks: From collective motion
  perspective,'' in \emph{2014 IEEE International Conference on Communications
  (ICC)}, 2014, pp. 203--208.

\bibitem{9063521}
M.~Liu, G.~Liao, N.~Zhao, H.~Song, and F.~Gong, ``Data-driven deep learning for
  signal classification in industrial cognitive radio networks,'' \emph{IEEE
  Transactions on Industrial Informatics}, vol.~17, no.~5, pp. 3412--3421,
  2021.

\bibitem{9520317}
X.~Liu, K.-Y. Lam, F.~Li, J.~Zhao, L.~Wang, and T.~S. Durrani, ``Spectrum
  sharing for 6g integrated satellite-terrestrial communication networks based
  on noma and cr,'' \emph{IEEE Network}, vol.~35, no.~4, pp. 28--34, 2021.

\bibitem{7160576}
P.~M. Rodriguez, I.~Val, A.~Lizeaga, and M.~Mendicute, ``Evaluation of
  cognitive radio for mission-critical and time-critical wsan in industrial
  environments under interference,'' in \emph{2015 IEEE World Conference on
  Factory Communication Systems (WFCS)}, 2015, pp. 1--4.

\bibitem{8254745}
P.~Si, H.~Liang, W.~Wu, and Y.~Zhang, ``Joint resource management in cognitive
  radio and edge computing based industrial wireless networks,'' in
  \emph{GLOBECOM 2017 - 2017 IEEE Global Communications Conference}, 2017, pp.
  1--6.

\bibitem{8887230}
S.~Demirci and D.~Gözüpek, ``Switching cost-aware joint frequency assignment
  and scheduling for industrial cognitive radio networks,'' \emph{IEEE
  Transactions on Industrial Informatics}, vol.~16, no.~7, pp. 4365--4377,
  2020.

\bibitem{7707414}
L.~Sibomana, H.-J. Zepernick, H.~Tran, and C.~Kabiri, ``A framework for packet
  delay analysis of point-to-multipoint underlay cognitive radio networks,''
  \emph{IEEE Transactions on Mobile Computing}, vol.~16, no.~9, pp. 2408--2421,
  2017.

\bibitem{7589870}
M.~E. Bayrakdar and A.~Çalhan, ``Delay characteristics of tdma medium access
  control protocol for cognitive radio networks,'' in \emph{2016 XIth
  International Scientific and Technical Conference Computer Sciences and
  Information Technologies (CSIT)}, 2016, pp. 66--69.

\bibitem{9306105}
M.~R. Amini and M.~W. Baidas, ``Performance analysis of cognitive-radio iot
  networks,'' in \emph{2020 IEEE Eighth International Conference on
  Communications and Networking (ComNet)}, 2020, pp. 1--7.

\bibitem{7746101}
P.~I. de~Almeida~Guimarães and J.~M. Câmara~Brito, ``Average delay in
  cognitive radio networks with imperfect sensing using slotted aloha
  protocol,'' in \emph{2016 International Symposium on Networks, Computers and
  Communications (ISNCC)}, 2016, pp. 1--6.

\bibitem{8785815}
Y.~Qin, H.~Zhang, and J.~Ren, ``Performance analysis of cognitive radio network
  for opportunistic wireless energy harvesting,'' in \emph{2019 IEEE 8th Joint
  International Information Technology and Artificial Intelligence Conference
  (ITAIC)}, 2019, pp. 700--707.

\bibitem{7297795}
F.~Lin, C.~Chen, N.~Zhang, X.~Guan, and X.~Shen, ``Autonomous channel
  switching: Towards efficient spectrum sharing for industrial wireless sensor
  networks,'' \emph{IEEE Internet of Things Journal}, vol.~3, no.~2, pp.
  231--243, 2016.

\bibitem{stoyan2013stochastic}
D.~Stoyan, W.~S. Kendall, S.~N. Chiu, and J.~Mecke, \emph{Stochastic geometry
  and its applications}.\hskip 1em plus 0.5em minus 0.4em\relax John Wiley \&
  Sons, 2013.

\bibitem{9115898}
M.~M. Azari, G.~Geraci, A.~Garcia-Rodriguez, and S.~Pollin, ``Uav-to-uav
  communications in cellular networks,'' \emph{IEEE Transactions on Wireless
  Communications}, vol.~19, no.~9, pp. 6130--6144, 2020.

\bibitem{7756327}
C.~Zhang and W.~Zhang, ``Spectrum sharing for drone networks,'' \emph{IEEE
  Journal on Selected Areas in Communications}, vol.~35, no.~1, pp. 136--144,
  2017.

\bibitem{7842290}
------, ``Spectrum sharing in drone small cells,'' in \emph{2016 IEEE Global
  Communications Conference (GLOBECOM)}, 2016, pp. 1--6.

\bibitem{9261468}
Z.~Wei, J.~Zhu, Z.~Guo, and F.~Ning, ``The performance analysis of spectrum
  sharing between uav enabled wireless mesh networks and ground networks,''
  \emph{IEEE Sensors Journal}, vol.~21, no.~5, pp. 7034--7045, 2021.

\bibitem{benvenuto2011principles}
N.~Benvenuto and M.~Zorzi, \emph{Principles of communications Networks and
  Systems}.\hskip 1em plus 0.5em minus 0.4em\relax John Wiley \& Sons, 2011.

\bibitem{3gpp.38.901}
3GPP, ``{3rd Generation Partnership Project; Technical Specification Group
  Radio Access Network; Study on channel model for frequencies from 0.5 to 100
  GHz },'' {3rd Generation Partnership Project (3GPP)}, Technical report (TR)
  38.901, 03 2022, v17.0.0.

\bibitem{3gpp.22.104}
------, ``{3rd Generation Partnership Project; Technical Specification Group
  Services and System Aspects; Service requirements for cyber-physical control
  applications in vertical domains},'' {3rd Generation Partnership Project
  (3GPP)}, Technical Specification (TS) 22.104, 12 2021, v18.3.0.

\bibitem{7835142}
S.~Agrawal, V.~Rana, and A.~K. Jagannatham, ``Queuing analysis for
  multiple-antenna cognitive radio wireless networks with beamforming,''
  \emph{IEEE Signal Processing Letters}, vol.~24, no.~3, pp. 334--338, 2017.

\bibitem{6478134}
F.~A. Khan, K.~Tourki, M.-S. Alouini, and K.~A. Qaraqe, ``Delay analysis of a
  point-to-multipoint spectrum sharing network with csi based power
  allocation,'' in \emph{2012 IEEE International Symposium on Dynamic Spectrum
  Access Networks}, 2012, pp. 235--241.

\bibitem{6562787}
------, ``Delay performance of a broadcast spectrum sharing network in
  nakagami- mm fading,'' \emph{IEEE Transactions on Vehicular Technology},
  vol.~63, no.~3, pp. 1350--1364, 2014.

\bibitem{tran2012delay}
H.~Tran, T.~Q. Duong, and H.-J. Zepernick, ``Delay performance of cognitive
  radio networks for point-to-point and point-to-multipoint communications,''
  \emph{EURASIP Journal on Wireless Communications and Networking}, vol. 2012,
  no.~1, pp. 1--14, 2012.

\bibitem{ross2014introduction}
S.~M. Ross, \emph{Introduction to probability models}.\hskip 1em plus 0.5em
  minus 0.4em\relax Academic press, 2014.

\bibitem{rayes2018internet}
A.~Rayes and S.~Salam, \emph{Internet of Things From Hype to Reality: The Road
  to Digitization}.\hskip 1em plus 0.5em minus 0.4em\relax Springer
  International Publishing, 2018.

\end{thebibliography}

\begin{appendices}
    \renewcommand{\thesection}{\Alph{section}}
    \section{Proof of Theorem 1}

According to (\ref{exp17}), the Cumulative Distribution Function (CDF) of the service delay can be expressed as (\ref{exp38}).
\begin{figure*}

\begin{equation}\label{exp38}
\begin{aligned}
{F_{{T_{mso}}}}( t )&{\rm{ =  }}P( {\frac{{{U_m}{N_m}}}{{{B_h}{{\log }_2}( {1{\rm{ + }}\frac{{P_m'y_0^{ - \alpha }{k_0}}}{{\sum\limits_{i \in {\phi _h}} {{P_h}x_i^{ - \alpha }{h_i}} {\rm{ + }}N\frac{{{B_h}}}{{{N_m}}}}}} )}} < t} )\\
\;\;\;\; &= \exp ( { - y_0^\alpha ( {{2^{\frac{{{U_m}{N_m}}}{{{B_h}t}}}} - 1} )\frac{{N{B_h}}}{{{P_m'}{N_m}}} - {\lambda _h}\frac{{2{\pi ^2}y_0^2{{( {\frac{{{P_h}}}{{P_m'}}( {{2^{\frac{{{U_m}{N_m}}}{{{B_h}t}}}} - 1} )} )}^{\frac{2}{\alpha }}}}}{{\alpha \sin ( {\frac{{2\pi }}{\alpha }} )}}}).
\end{aligned}
\end{equation}
{\noindent} \rule[-10pt]{18cm}{0.01em}
\end{figure*}
The service delay of a data packet is finite. If the service delay is greater than a threshold $t_{out}$, it is regarded as a transmission failure. Hence, the probability of packet transmission failure is expressed as (\ref{exp39}).
\begin{figure*}
\begin{equation}\label{exp39}
{P_{mo\_out}}{\rm{ = 1 - }}\exp ( { - y_0^\alpha ( {{2^{\frac{{{U_m}{N_m}}}{{{B_m}{t_{out}}}}}} - 1} )\frac{{N{B_h}}}{{P_m'{N_m}}} - {\lambda _h}\frac{{2{\pi ^2}y_0^2{{( {\frac{{{P_h}}}{{P_m'}}( {{2^{\frac{{{U_m}{N_m}}}{{{B_h}{t_{out}}}}}} - 1} )})}^{\frac{2}{\alpha }}}}}{{\alpha \sin ( {\frac{{2\pi }}{\alpha }})}}}).
\end{equation}
{\noindent} \rule[-10pt]{18cm}{0.01em}
\end{figure*}
The moments of service delay can be expressed as \cite{7835142,6478134,6562787,tran2012delay}
\begin{equation}\label{exp40}
\begin{aligned}
E( {T_{mso}^i} ) = &E[ {T_{mso}^i|{T_{mso}} < {t_{out}}} ]P( {{T_{mso}} < {t_{out}}} )\\
&+ E[ {T_{mso}^i|{T_{mso}} \ge {t_{out}}} ]P( {{T_{mso}} \ge {t_{out}}} )\\
= &E[ {T_{mso}^i|{T_{mso}} < {t_{out}}} ]{F_{{T_{mso}}}}t( {{t_{out}}} ) \\
&+ t_{out}^i( {1 - {F_{{T_{mso}}}}( {{t_{out}}} )} ),
\end{aligned}
\end{equation}
where $i = 1$ is the first moment of service delay, and $i = 2$ is the second moment of service delay.

\indent When the service delay of the data packet is smaller than a threshold, namely ${T_{mso}} < {t_{out}}$, the data packet of MTC network is successfully transmitted to the typical MTC device. The service delay of the data packet when ${T_{mso}} < {t_{out}}$ is represented by ${\overline T _{ms}}$, and its CDF is
\begin{equation}\label{exp41}
\begin{aligned}
P( {{{\overline T }_{mso}}|{T_{mso}} < {t_{out}}} ) &= \frac{{P( {{{\overline T }_{mso}},{T_{mso}} < {t_{out}}} )}}{{P( {{T_{mso}} < {t_{out}}} )}} \\
&= \frac{{P( {{{\overline T }_{mso}},{T_{mso}} < {t_{out}}} )}}{{1 - {P_{mo-out}}}}.
\end{aligned}
\end{equation}
\indent The Probability Density Function (PDF) of ${\overline T _{mso}}$ is \cite{7835142,6478134,6562787,tran2012delay}
\begin{equation}\label{exp42}
{f_{{{\overline T }_{mso}}}}( t ) = \frac{{{f_{{T_{mso}}}}( t )}}{{1 - {P_{mo-out}}}}, {\rm \rm 0} \le t \le {t_{out}}.
\end{equation}
\indent When the packet transmission is successful, the expectation of service delay can be expressed as
\begin{equation}\label{exp43}
\begin{aligned}
&E[ {{T_{mso}}|{T_{mso}} < {t_{out}}} ]\\
&{\rm{ = }}\int_0^{{t_{out}}} {t{f_{\mathop {{T_{mso}}}\limits^ -  }}( t )} dt\\
&{\rm{ = }}\frac{1}{{1 - {P_{mo-out}}}}( {{t_{out}}{F_{{T_{mso}}}}( {{t_{out}}} ) - \int_0^{{t_{out}}} {{F_{{T_{mso}}}}( t )dt} } ).
\end{aligned}
\end{equation}
\indent Similar to the derivation of (\ref{exp43}), the second moment of the service delay when the packet is successfully transmitted is
\begin{equation}\label{exp44}
\begin{aligned}
&E[ {T_{mso}^2|{T_{mso}} < {t_{out}}} ]\\
&{\rm{ = }}\int_0^{{t_{out}}} {{t^2}{f_{_{{{\overline T }_{mso}}}}}( t )} dt\\
&{\rm{ = }}\frac{1}{{1 - {P_{mo-out}}}}( {t_{out}^2{F_{{T_{mso}}}}( {{t_{out}}} ) - 2\int_0^{{t_{out}}} {t{F_{{T_{mso}}}}( t )dt} } ).
\end{aligned}
\end{equation}
\indent Substituting (\ref{exp43}) into (\ref{exp40}), the first moment of service delay is
\begin{equation}\label{exp45}
E( {{T_{mso}}} ) = {t_{out}} - \int_0^{{t_{out}}} {{F_{{T_{mso}}}}( t)} dt.
\end{equation}
\indent Similar to the derivation of (\ref{exp45}), the second moment of the service delay is
\begin{equation}\label{exp46}
E( {T_{mso}^2} ) = t_{out}^2 - 2\int_0^{{t_{out}}} {t{F_{{T_{mso}}}}( t )} dt.
\end{equation}
\indent In the M/G/1 queueing model, the waiting delay of a data packet is \cite{ross2014introduction}
\begin{equation}\label{exp47}
E( {{T_{mwo}}} ) = \frac{{{\lambda _{md}}E( {T_{mso}^2} )}}{{2(1 - {\rho _m})}},
\end{equation}
where  ${\rho _m}$ is the load intensity, which is expressed as
\begin{equation}\label{exp48}
{\rho _m}  = {\lambda _{md}}E({{T_{mso}}}).
\end{equation}
\indent According to (\ref{exp45}), (\ref{exp47}) and (\ref{exp48}), the average waiting delay is
\begin{equation}\label{exp49}
E( {{T_{mwo}}} ) = \frac{{{\lambda _{md}}( {t_{out}^2 - 2\int_0^{{t_{out}}} {t{F_{{T_{mso}}}}( t )} dt} )}}{{2( {1 - {\lambda _{md}}( {{t_{out}} - \int_0^{{t_{out}}} {{F_{{T_{mso}}}}( t )} dt} )} )}}.
\end{equation}
\indent Delay jitter is defined as the variation of the delay \cite{rayes2018internet}. When MTC network uses proprietary spectrum, the expression of delay jitter is
\begin{equation}\label{exp50}
{J_{mo}} = D( {{T_{mso}}} ) + D( {{T_{mwo}}}).
\end{equation}
\indent According to (\ref{exp45}) and (\ref{exp46}), the variance of service delay can be expressed as
\begin{equation}\label{exp51}
\begin{aligned}
D( {{T_{mso}}}) =& t_{out}^2 - 2\int_0^{{t_{out}}} {t{F_{{T_{mso}}}}( t )} dt\\
&- {( {{t_{out}} - \int_0^{{t_{out}}} {{F_{{T_{mso}}}}( t )} dt})^2}.
\end{aligned}
\end{equation}
\indent The variance of waiting delay is expressed as \cite{ross2014introduction}
\begin{equation}\label{exp52}
\begin{aligned}
D( {{T_{mwo}}} ){\rm{ = }}&{[ {E( {{T_{mwo}}} )} ]^2}+ \frac{{\lambda E( {T_{mso}^3} )}}{{3( {1 - \rho_m } )}}.
\end{aligned}
\end{equation}
\indent According to (\ref{exp40}), ${E( {T_{mso}^3} )}$ is
\begin{equation}\label{exp53}
\begin{array}{c}
E( {T_{mso}^3} ) = t_{{t_{out}}}^3 - 3\int_0^{{t_{out}}} {{t^2}{F_{{T_{mso}}}}( t)} dt.
\end{array}
\end{equation}
\indent Therefore, (\ref{exp52}) is rewritten as
\begin{equation}\label{exp54}
\begin{aligned}
D( {{T_{mwo}}} ) &= {( {\frac{{{\lambda _{md}}( {t_{out}^2 - 2\int_0^{{t_{out}}} {t{F_{{T_{mso}}}}( t)} dt} )}}{{2( {1 - {\lambda _{md}}( {{t_{out}} - \int_0^{{t_{out}}} {{F_{{T_{mso}}}}( t )} dt} )} )}}} )^2} \\
&+ \frac{{{\lambda _{md}}( {t_{{t_{out}}}^3 - 3\int_0^{{t_{out}}} {{t^2}{F_{{T_{mso}}}}( t )} dt} )}}{{2( {1 - {\lambda _{md}}( {{t_{out}} - \int_0^{{t_{out}}} {{F_{{T_{mso}}}}( t)} dt} )} )}}.
\end{aligned}
\end{equation}
\indent According to (\ref{exp51}) and (\ref{exp54}), the delay jitter in Theorem 1 is derived.

\section{Proof of Theorem 2}

According to (\ref{exp25}), the CDF of the service delay is
\begin{equation}\label{exp55}
\begin{aligned}
{F_{{T_{ms}}}}( t )&{\rm{ =  }}P( {\frac{{{U_m}{N_m}}}{{{B_m}{{\log }_2}( {1 + \gamma _m} )}} < t} )\\
& = P( {\gamma _m > {2^{\frac{{{U_m}{N_m}}}{{{B_m}t}}}} - 1} )\\
& = \exp ( { - y_0^\alpha ( {{2^{\frac{{{U_m}{N_m}}}{{{B_m}t}}}} - 1} )\frac{{N{B_m}}}{{{P_m}{N_m}}}}).
\end{aligned}
\end{equation}
\indent Hence, the probability of packet transmission failure is\\
\begin{equation}\label{exp56}
{P_{m-out}} = 1 - {\rm{exp}}( { - y_0^\alpha ( {{2^{\frac{{{U_m}{N_m}}}{{{B_m}{t_{out}}}}}} - 1} )\frac{{N{B_m}}}{{{P_m}{N_m}}}} ).
\end{equation}
\indent Since it is similar to the proof of theorem 1, the proof of theorem 2 will not be described in detail.

\section{Proof of Theorem 3}

According to (\ref{exp33}), the CDF of the service delay of the data packet is expressed as (\ref{exp57}),
where ${F_u}\left( z \right)$ can be expressed as (\ref{exp58}). Letting
\begin{figure*}
\begin{equation}\label{exp57}
\begin{array}{l}
F_{{T_{ms}}}'( t ){\rm{ =  }}P( {\frac{{{U_m}{N_m}}}{{{B_h}{{\log }_2}( {1{\rm{ + }}\frac{{{P_m'}y_0^{ - \alpha }{k_0}}}{{\sum\limits_{i \in {\phi _h}} {{P_h}x_i^{ - \alpha }{h_i}}  + \frac{{N{B_h}}}{{{N_m}}}}}} ) + {B_m}{{\log }_2}( {1{\rm{ + }}\frac{{{P_m}{N_m}y_0^{ - \alpha }{g_0}}}{{N{B_m}}}})}} < t} )\\
\;\;\;\;\;\;\;\;\;\;\;\;{\rm{ = }}1 - P( {{B_h}{{\log }_2}( {1{\rm{ + }}\frac{{{P_m'}y_0^{ - \alpha }{k_0}}}{{\sum\limits_{i \in {\phi _h}} {{P_h}x_i^{ - \alpha }{h_i}}  + \frac{{N{B_h}}}{{{N_m}}}}}}) + {B_m}{{\log }_2}( {1{\rm{ + }}\frac{{{P_m}{N_m}y_0^{ - \alpha }{g_0}}}{{N{B_m}}}}) \le \frac{{{U_m}{N_m}}}{t}} )\\
\;\;\;\;\;\;\;\;\;\;\;\; = 1 - {F_u}( z ){|_{z = \frac{{{U_m}{N_m}}}{t}}},
\end{array}
\end{equation}
{\noindent} \rule[-10pt]{18cm}{0.01em}
\end{figure*}
\begin{figure*}
\begin{equation}\label{exp58}
{F_u}(z) = P( {{B_h}{{\log }_2}( {1{\rm{ + }}\frac{{{P_m'}y_0^{ - \alpha }{k_0}}}{{\sum\limits_{i \in {\phi _h}} {{P_h}x_i^{ - \alpha }{h_i}}  + \frac{{N{B_h}}}{{{N_m}}}}}}) + {B_m}{{\log }_2}( {1{\rm{ + }}\frac{{{P_m}{N_m}y_0^{ - \alpha }{g_0}}}{{N{B_m}}}}) \le {z}} ).
\end{equation}
\end{figure*}
\begin{small}
\begin{equation}\label{exp59}
{F_1}( \tau ) = P( {{B_h}{{\log }_2}( {1{\rm{ + }}\frac{{{P_m}'y_0^{ - \alpha }{k_0}}}{{\sum\limits_{i \in {\phi _h}} {{P_h}x_i^{ - \alpha }{h_i}} {\rm{ + }}\frac{{N{B_h}}}{{{N_m}}}}}}) \le \tau }),
\end{equation}
\end{small}
similar to the derivation of Lemma 1, (\ref{exp60}) is derived. Letting
\begin{figure*}
{\noindent} \rule[-10pt]{18cm}{0.01em}
\begin{equation}\label{exp60}
{F_1}( \tau ) = 1 - \exp ( { - y_0^\alpha( {{2^{\frac{\tau }{{{B_h}}}}} - 1} )\frac{{N{B_h}}}{{{P_m'}{N_m}}} - {\lambda _h}\frac{{2{\pi ^2}y_0^2{{( {\frac{{{P_h}}}{{P_m'}}( {{2^{\frac{\tau }{{{B_h}}}}} - 1} )} )}^{\frac{2}{\alpha }}}}}{{\alpha \sin ( {\frac{{2\pi }}{\alpha }})}}}),
\end{equation}
{\noindent} \rule[-10pt]{18cm}{0.01em}
\end{figure*}
\begin{small}
\begin{equation}\label{exp61}
{F_2}( \tau) = P( {{B_m}{{\log }_2}( {1{\rm{ + }}\frac{{{P_m}{N_m}y_0^{ - \alpha }{{\rm{g}}_0}}}{{N{B_m}}}} ) < \tau} ),
\end{equation}
\end{small}
then
\begin{equation}\label{exp62}
{F_2}( \tau ){\rm{ = }}1 - \exp ( { - y_0^\alpha ( {{2^{\frac{\tau}{{{B_m}}}}} - 1} )\frac{{N{B_m}}}{{{P_m}{N_m}}}} ).
\end{equation}
\indent The derivative of (\ref{exp62}) with respect to $\tau$ is
\begin{small}
\begin{equation}\label{exp63}
{f_2}(\tau ) = {\rm{exp}}( { - y_0^\alpha ( {{2^{\frac{\tau}{{{B_m}}}}} - 1} )\frac{{N{B_m}}}{{{P_m}{N_m}}}} )( {\frac{{{2^{\frac{\tau}{{{B_m}}}}}N{{\log }_e}( 2 )y_0^\alpha }}{{{P_m}{N_m}}}} ).
\end{equation}
\end{small}
\indent Since ${F_1}(\tau )$ and ${F_2}(\tau)$ are independent, the PDF of ${F_u}( z )$ is
\begin{equation}\label{exp64}
\begin{aligned}
{f_u}\left( z \right) = {f_1}\left( z \right)*{f_2}\left( z \right) = \int_0^t {{f_1}\left( \tau  \right)}  \cdot {f_2}\left( {z - \tau } \right)d\tau,
\end{aligned}
\end{equation}
where $*$ represents performing convolution operation. Using the convolution integral property
\begin{equation}\label{exp65}
\begin{array}{l}
\int_{ - \infty }^z {{f_1}(z) * {f_2}(z} )dz\\
 = (\int_{ - \infty }^z {{f_1}(z)} dz) * {f_2}(z)\\
 = (\int_{ - \infty }^z {{f_2}(z)} dz) * {f_1}(z),
\end{array}
\end{equation}
${F_u}( z )$ can be expressed as
\begin{equation}\label{exp66}
\begin{aligned}
{F_u}( z ) = &\int_{ - \infty }^z {f( z )} dz\\
{\rm{ = }}&\int_0^z {\; {\int_0^t {{f_1}( \tau ) \cdot {f_2}( {z - \tau} )d\tau} } } dz\\
=& {F_1}( z )*{f_2}( z )\\
{\rm{ = }}&\int_0^z {\;{F_1}( \tau  ) \cdot {f_2}( {z - \tau } )} d\tau.
\end{aligned}
\end{equation}
\indent The CDF of the service delay is
\begin{equation}\label{exp67}
F_{_{{T_{ms}}}}'( t )\;  = 1 - \int_0^{\frac{{{U_m}{N_m}}}{t}} {\;{F_1}( \tau ) \cdot {f_2}( {\frac{{{U_m}{N_m}}}{t} - \tau })} d\tau.
\end{equation}
\indent The probability of packet transmission failure is
\begin{equation}\label{exp68}
P_{m-out}'=\int_0^{\frac{{{U_m}{N_m}}}{{{t_{out}}}}} {\;{F_1}( \tau  ) \cdot {f_2}( {\frac{{{U_m}{N_m}}}{{{t_{out}}}} - \tau } )} d\tau.
\end{equation}
\begin{figure*}
\begin{equation}\label{exp69}
\begin{aligned}
&E\left[ {T_{ms}'|{T_{ms}} < {t_{out}}} \right]\\
&{\rm{ = }}\frac{1}{{1 - P_{m-out}'}}( {{t_{out}}{F_{T_{ms}'}}( {{t_{out}}} ) - \int_0^{{t_{out}}} {{F_{T_{ms}'}}( t )dt} } )\\
& = \frac{1}{{1 - P_{m-out}'}}( {{t_{out}}{F_{T_{ms}'}}( {{t_{out}}} ) - t_{out} + \int_0^{{t_{out}}} { {\int_0^{\frac{{{U_m}{N_m}}}{t}} {\;{F_1}( \tau ) \cdot {f_2}( {\frac{{{U_m}{N_m}}}{t} - \tau } )} d\tau } dt} } )\;\;\;\\
&{\rm{ = }}\frac{1}{{1 - P_{m-out}'}}( {\int_0^{{t_{out}}} {{\int_0^{\frac{{{U_m}{N_m}}}{t}} {\;{F_1}( \tau  ) \cdot {f_2}( {\frac{{{U_m}{N_m}}}{t} - \tau } )} d\tau } dt}  - t_{out} \cdot {F_u}( {\frac{{{U_m}{N_m}}}{{t_{out}}}} )} ).
\end{aligned}
\end{equation}
{\noindent} \rule[-10pt]{18cm}{0.01em}
\end{figure*}
\indent Using partial integration, (\ref{exp69}) is derived.

\indent Similar to (\ref{exp69}), (\ref{exp70}) is obtained as
\begin{equation}\label{exp70}
\begin{aligned}
E\left[ {T_{ms}^{2'}|T_{ms}' < {t_{out}}} \right]{\rm{ = }}&\frac{1}{{1 - P_{m-out}'}}( t_{out}^2{F_{T_{ms}'}}( {{t_{out}}} )\\
&- 2\int_0^{{t_{out}}} {t{F_{T_{ms}'}}( t )dt} ).
\end{aligned}
\end{equation}
\indent Substituting (\ref{exp68}) and (\ref{exp69}) into (\ref{exp40}), the average service delay is obtained as
\begin{equation}\label{exp71}
\begin{aligned}
E\left( {T_{ms}'} \right) &= \int_0^{{t_{out}}} {{F_u}( {\frac{{{U_m}{N_m}}}{t}} )dt} \\
&= \int_0^{{t_{out}}} { {\int_0^{\frac{{{U_m}{N_m}}}{t}} {\;{F_1}( \tau ) \cdot {f_2}( {\frac{{{U_m}{N_m}}}{t} - \tau } )} d\tau } dt}.
\end{aligned}
\end{equation}
\indent Similarly, the second moment of the service delay is
\begin{small}
\begin{equation}\label{exp72}
\begin{aligned}
E\left( {T_{ms}^{2'}} \right) &= 2\int_0^{{t_{out}}} {t{F_u}( {\frac{{{U_m}{N_m}}}{t}} )} dt\\
&{\rm{ = 2}}\int_0^{{t_{out}}} { {\int_0^{\frac{{{U_m}{N_m}}}{t}} {\;t{F_1}( \tau ) \cdot {f_2}( {\frac{{{U_m}{N_m}}}{t} - \tau })} d\tau } dt}.
\end{aligned}
\end{equation}
\end{small}
\begin{spacing}{1.0}
\indent  (\ref{exp34}) in Theorem 3 is derived by (\ref{exp47}), (\ref{exp71}) and (\ref{exp72}). Substituting (\ref{exp67}-\ref{exp68}) into (\ref{exp40}), the third moment of the service delay can be derived as
\end{spacing}
\begin{equation}\label{exp73}
\begin{aligned}
E\left[ {T_{ms}^{3'}|T_{ms}' < {t_{out}}} \right]{\rm{ = }}&\frac{1}{{1 - P_{m-out}'}}( t_{out}^3{F_{T_{ms}'}}( {{t_{out}}} )\\
&- 3\int_0^{{t_{out}}} {{t^2}{F_{T_{ms}'}}( t )dt}),
\end{aligned}
\end{equation}
and
\begin{footnotesize}
\begin{equation}\label{exp74}
\begin{aligned}
E\left( {T_{ms}^{3'}} \right) &= 3\int_0^{{t_{out}}} {{t^2}} {F_u}( {\frac{{{U_m}{N_m}}}{t}} )dt\\
 &= {\rm{3}}\int_0^{{t_{out}}} { {\int_0^{\frac{{{U_m}{N_m}}}{t}} {\;{t^2}{F_1}( \tau  ) \cdot {f_2}( {\frac{{{U_m}{N_m}}}{t} - \tau })} d\tau } dt}.
\end{aligned}
\end{equation}
\end{footnotesize}
\indent Substituting (\ref{exp71}-\ref{exp74}) and (\ref{exp50}) into (\ref{exp51}-\ref{exp54}), (\ref{exp32}) in Theorem 3 is derived.

\end{appendices}

\begin{IEEEbiography}[{\includegraphics[width=1in,height=1.25in,clip,keepaspectratio]{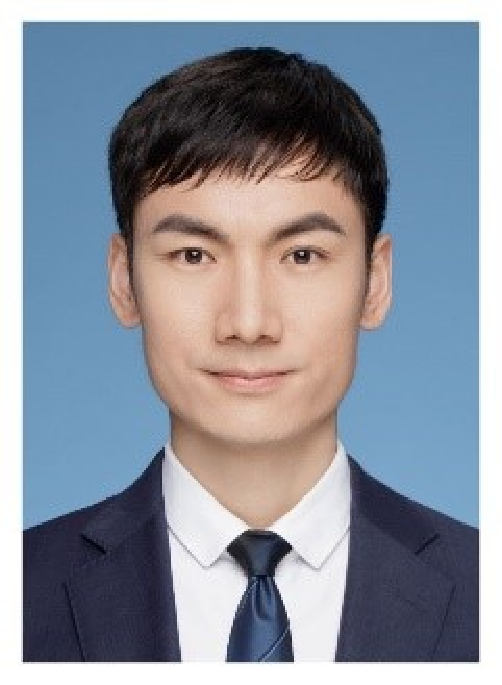}}]
{Zhiqing Wei} (Member, IEEE) received the B.E. and Ph.D. degrees from the Beijing University of Posts and Telecommunications (BUPT), Beijing, China, in 2010 and 2015, respectively. He is an Associate Professor with BUPT. He has authored one book, three book chapters, and more than 50 papers. His research interest is the performance analysis and optimization of intelligent machine networks. He was granted the Exemplary Reviewer of IEEE WIRELESS COMMUNICATIONS LETTERS in 2017, the Best Paper Award of WCSP 2018. He was the Registration Co-Chair of IEEE/CIC ICCC 2018, the publication Co-Chair of IEEE/CIC ICCC 2019 and IEEE/CIC ICCC 2020.
\end{IEEEbiography}

\begin{IEEEbiography}[{\includegraphics[width=1in,height=1.25in,clip,keepaspectratio]{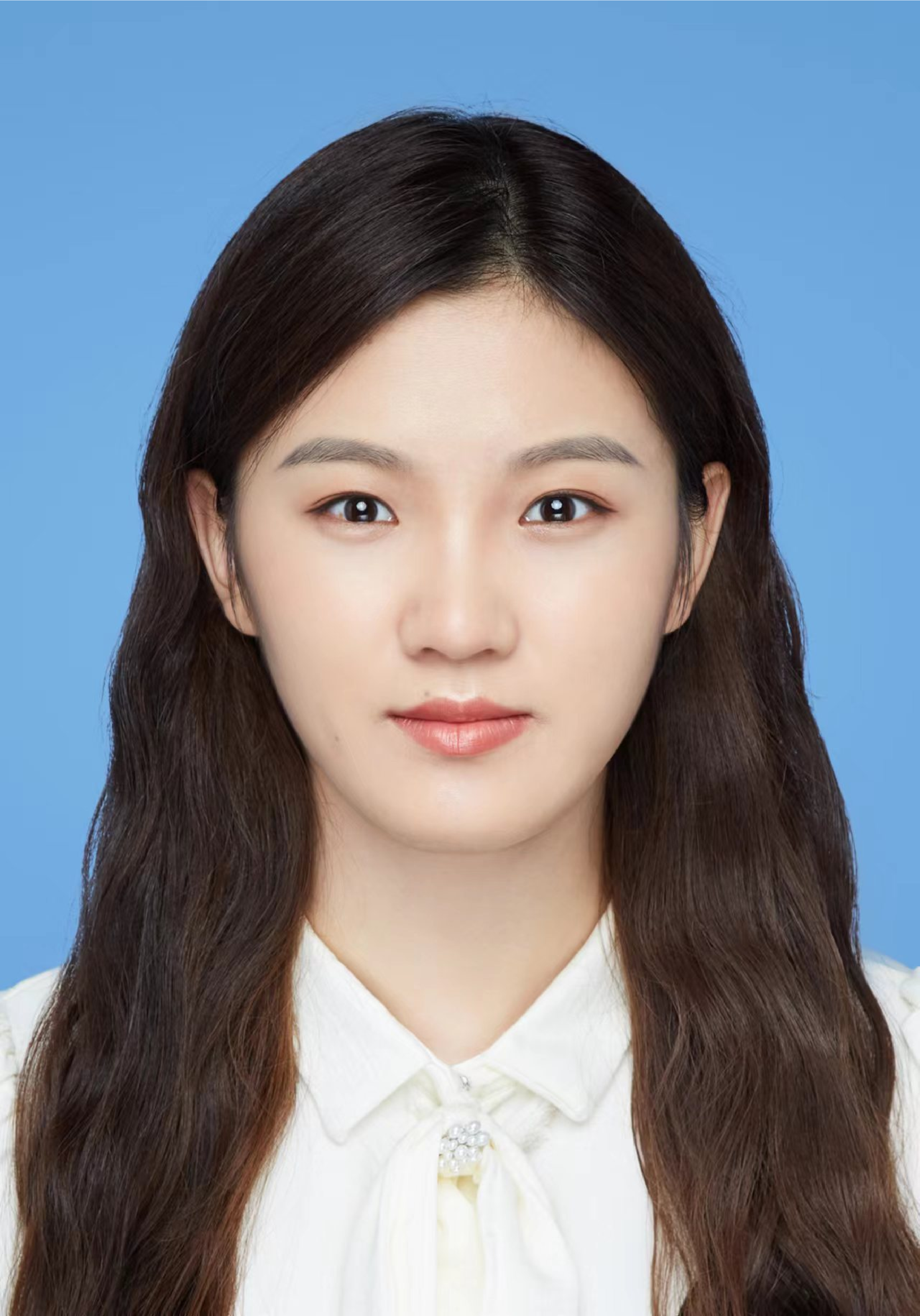}}]
{Ling Zhang} (Student Member, IEEE) received the B.E. degree from China University of Petroleum (UPC), Shandong, China, in 2020. She is working on the master's degree in Beijing University of Posts and Telecommunications (BUPT). Her research interest is the performance analysis of intelligent machine networks.
\end{IEEEbiography}

\begin{IEEEbiography}[{\includegraphics[width=1in,height=1.25in,clip,keepaspectratio]{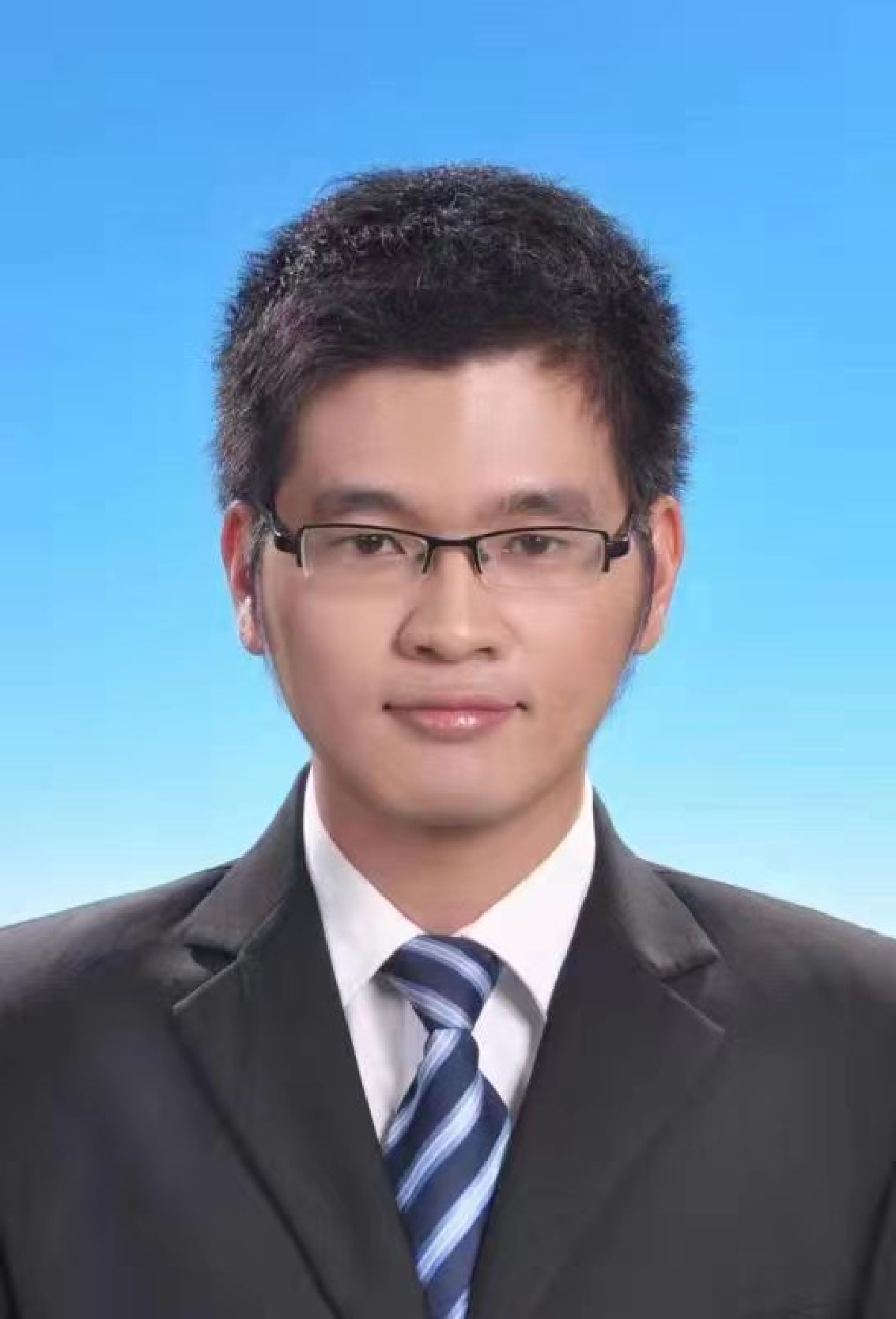}}]
{Gaofeng Nie}(Member, IEEE) received the B.S. degree in communications engineering and the Ph.D. degree in telecommunications and information system from the Beijing University of Posts and Telecommunications (BUPT), in 2010 and 2016, respectively. He is currently a Lecturer with BUPT. His research interests include SDN over wireless networks and key technologies in B5G/6G wireless networks.
\end{IEEEbiography}

\begin{IEEEbiography}[{\includegraphics[width=1in,height=1.25in,clip,keepaspectratio]{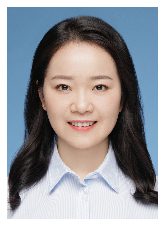}}]
{Huici Wu} (Member, IEEE) received the Ph.D degree from Beijing University of Posts and Telecommunications (BUPT), Beijing, China, in 2018. From 2016 to 2017, she visited the Broadband Communications Research (BBCR) Group, University of Waterloo, Waterloo, ON, Canada. She is now an Associate Professor at BUPT. Her research interests are in the area of wireless communications and networks, with current emphasis on collaborative air-to-ground communication and wireless access security.
\end{IEEEbiography}

\begin{IEEEbiography}[{\includegraphics[width=1.1in,height=1.25in,clip,keepaspectratio]{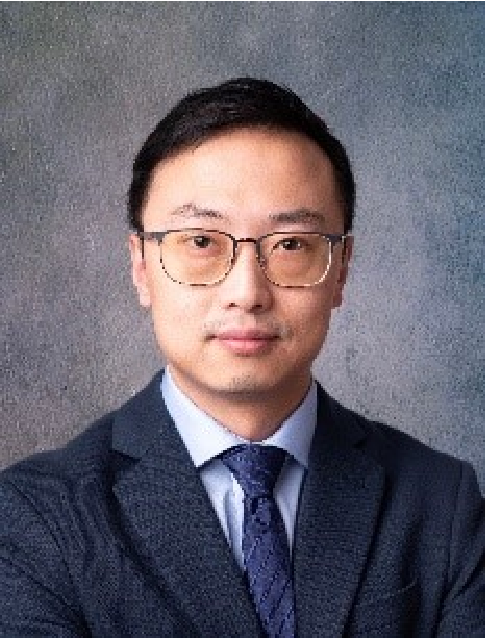}}]{Ning Zhang} (Senior Member, IEEE)  received the Ph.D degree in Electrical and Computer Engineering from University of Waterloo, Canada, in 2015. After that, he was a postdoc research fellow at University of Waterloo and University of Toronto, respectively. Since 2020, he has been an Associate Professor in the Department of Electrical and Computer Engineering at University of Windsor, Canada. His research interests include connected vehicles, mobile edge computing, wireless networking, and security. He is a Highly Cited Researcher (Web of Science). He serves/served as an Associate Editor of IEEE Transactions on Mobile Computing, IEEE Communications Surveys and Tutorials, IEEE Internet of Things Journal, and IEEE Transactions on Cognitive Communications and Networking. He also serves/served as a TPC chair for IEEE VTC 2021 and IEEE SAGC 2020, a general chair for IEEE SAGC 2021, a chair for track of several international conferences and workshops including IEEE ICC, VTC, INFOCOM Workshop, and Mobicom Workshop. He received a number of Best Paper Awards from conferences and journals, such as IEEE Globecom, IEEE ICC, IEEE ICCC, IEEE WCSP, and Journal of Communications and Information Networks. He also received IEEE TCSVC Rising Star Award and IEEE ComSoc Young Professionals Outstanding Nominee Award. He serves as the Vice Chair for IEEE Technical Committee on Cognitive Networks and IEEE Technical Committee on Big Data.
\end{IEEEbiography}

\begin{IEEEbiography}[{\includegraphics[width=1in,height=1.25in,clip,keepaspectratio]{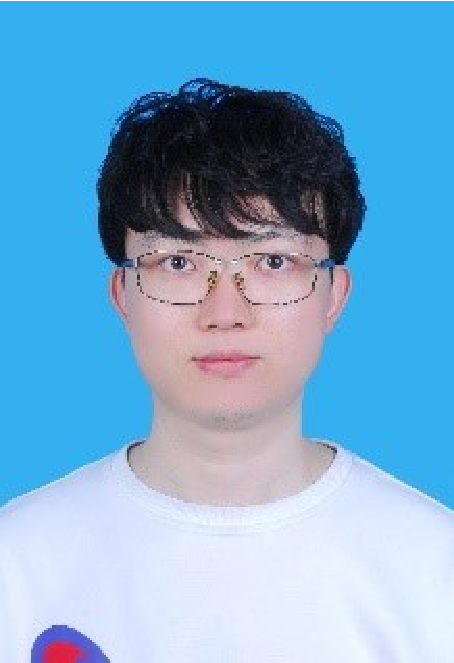}}]
{Zeyang Meng} (Student Member, IEEE)) received the B.E. degree from Beijing University of Posts and Telecommunications (BUPT), Beijing, China, in 2020. He is working on the Ph.D. degree in BUPT. His research interest is the performance analysis of intelligent machine networks.
\end{IEEEbiography}

\begin{IEEEbiography}[{\includegraphics[width=1.1in,height=1.25in,clip,keepaspectratio]{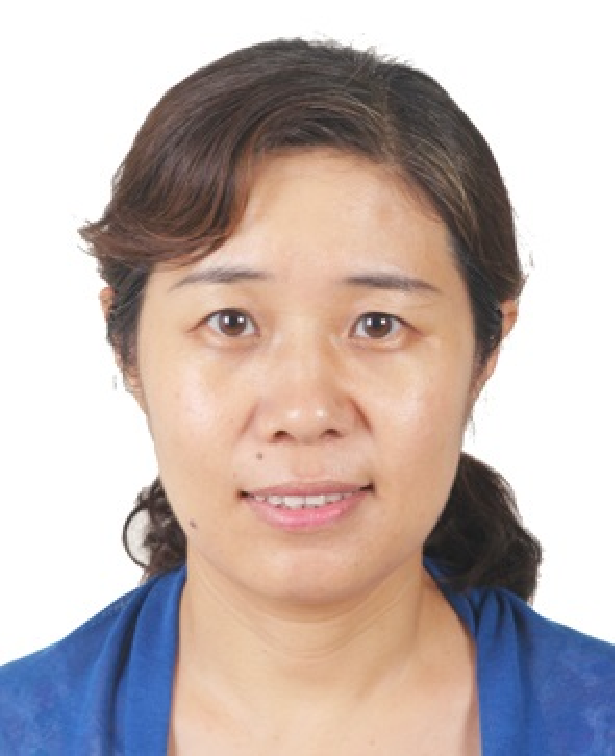}}]{Zhiyong Feng} (M'08-SM'15) received her B.E., M.E., and Ph.D. degrees from Beijing University of Posts and Telecommunications (BUPT), Beijing, China. She is a professor at BUPT, and the director of the Key Laboratory of the Universal Wireless Communications, Ministry of Education, P.R.China. She is a senior member of IEEE, vice chair of the Information and Communication Test Committee of the Chinese Institute of Communications (CIC). Currently, she is serving as Associate Editors-in-Chief for China Communications, and she is a technological advisor for international forum on NGMN. Her main research interests include wireless network architecture design and radio resource management in 5th generation mobile networks (5G), spectrum sensing and dynamic spectrum management in cognitive wireless networks, and universal signal detection and identification.
\end{IEEEbiography}

\vfill

\end{document}